\documentclass[11pt]{article}[12pt]
\usepackage[utf8]{inputenc}
\usepackage[english]{babel}
\usepackage{amsmath}
\usepackage{amssymb}
\usepackage{amsthm}
\usepackage{graphicx}
\usepackage{epstopdf}
\usepackage[affil-it]{authblk}
\usepackage{upgreek}
\usepackage{subfigure}
\usepackage{xcolor}

\usepackage{siunitx}

\usepackage{hyperref}
\hypersetup{
    colorlinks=true,
    linkcolor=red,
    citecolor=blue,
    urlcolor=blue
    }

\usepackage{setspace}
    
\usepackage[a4paper, total={6in, 9.6in}]{geometry}

\def\bSig\mathbf{\Sigma}

\newcommand{\bbeta}{{\mbox{\boldmath$\beta$}}}
\newcommand{\ddL}{{\mbox{\boldmath$\ddot{L}$}}}
\newcommand{\Deltta}{{\mbox{\boldmath$\Delta$}}}
\newcommand{\ele}{{\mbox{\boldmath$l$}}}
\newcommand{\F}{{\mbox{\boldmath$F$}}} 
\newcommand{\hatbbeta}{{\mbox{\boldmath$\hat{\beta}$}}}
\newcommand{\hatttheta}{{\mbox{\boldmath$\hat{\theta}$}}}

\newcommand{\ommega}{{\mbox{\boldmath$\omega$}}}

\newcommand{\R}{\mathbb{R}}   
\newcommand{\senh}{\mathrm{senh}}

\newcommand{\ttheta}{{\mbox{\boldmath$\theta$}}}

\newcommand{\U}{{\mbox{\boldmath$U$}}}

\newcommand{\X}{{\mbox{\boldmath$X$}}}

\newtheorem{proposition}{Proposition}

\title{\textbf{Stability  Analysis and Local Influence Diagnostics for an Extreme-Value Regression Model of Anomalous Wind Gusts}}

\author[1,3]{José I. C. Lima}
\author[2,3]{Raydonal Ospina}
\author[1]{Michelli Barros}
\author[4]{Antônio M. S. Macêdo}

\affil[1]{\it  Unidade Acadêmica de Estatística, Universidade Federal de Campina Grande,  58429-900, Campina Grande/PB, Brazil}

\affil[2]{\it Departamento de Estatística, LInCa, Universidade Federal da Bahia, 40170--110, Salvador/BA, Brazil}

\affil[3]{\it Departamento de Estatística, CASTLab, Universidade Federal de Pernambuco, 50740--540, Recife/PE, Brazil}

\affil[4]{\it  Laboratório de Física Teórica e Computacional, Departamento de Física, Universidade Federal de Pernambuco, 50670-901 Recife/PE, Brazil \\ \vspace{1cm}}

\bigskip 

\affil[1]{\tt jose.iraponil@ufcg.edu.br}
\affil[2]{\tt michellikarinne@uaest.ufcg.edu.br}
\affil[3]{\tt raydonal@de.ufpe.br}
\affil[4]{\tt amsmacedo@df.ufpe.br}

\date{}

\begin{document}

\maketitle

\begin{abstract}
Extreme events in complex physical systems, such as anomalous wind gusts, often cause significant material and human damage. Their modeling is crucial for risk assessment and understanding the underlying dynamics. In this work, we introduce a local influence analysis to assess the stability of a class of extreme-value Birnbaum-Saunders regression models, which are particularly suited for analyzing such data. The proposed approach uses the conformal normal curvature (CNC) of the log-likelihood function to diagnose the influence of individual observations on the postulated model. By examining the eigenvalues and eigenvectors associated with the CNC, we identify influential data points—physical events that disproportionately affect the model's parameters. We illustrate the methodology through a simulation study and apply it to a time series of wind gust data from Itajaí, Brazil, where a severe event caused multiple damages and casualties. Our approach successfully pinpoints this specific event as a highly influential observation and quantifies its impact on the fitted model. This work provides a valuable diagnostic tool for physicists and data scientists working with extreme-value models of complex natural phenomena.
\end{abstract}


\section{Introduction}

Extreme events are ubiquitous across physical systems and scales, ranging from noise-seeded extreme intensity fluctuations in optical fibers  and cavities \cite{solli2007,dudley2008,lima2017,hammani2008,dudley2014,onorato2013} and multifractal fluctuations in various complex systems \cite{Muzy-1,Muzy-2} to catastrophic atmospheric disturbances such as anomalous wind gusts, tropical cyclones, and extreme temperatures \cite{hueso2006,lobeto2021,lam2023} and infamous oceanic rogue waves \cite{onorato2001,janssen2003,chabchoub2011,onorato2013}

These rare but consequential phenomena manifest as optical rogue waves in fiber lasers \cite{lima2017}, anomalous wind gusts that devastate infrastructure \cite{lobeto2021}, freak ocean waves that threaten maritime operations \cite{onorato2001}, and even disruptions in biological systems such as DNA replication failures \cite{bechhoefer2007}. Despite their diverse origins, these extreme events share a fundamental characteristic: they represent dramatic deviations from typical system behavior that carry disproportionate physical significance \cite{sornette2006critical}.

The statistical characterization of extreme events presents profound challenges that transcend traditional approaches based on Gaussian assumptions. In turbulent flows, the probability density functions of velocity increments exhibit pronounced skewness and fat tails, reflecting the intermittent nature of energy dissipation across scales \cite{sosa2019}. Solar wind parameters display Fr\'echet-type extreme value statistics, indicating power-law tails that fundamentally alter risk assessments for space weather \cite{Maloney_2010, Schumann_2012}. Even in controlled laboratory settings, random fiber lasers transition from Gaussian to L\'evy statistics near threshold, demonstrating how nonlinear dynamics can dramatically reshape statistical behavior \cite{lima2017}. These observations demand a unified theoretical framework capable of capturing both the universality and system-specific features of extreme phenomena.

Extreme Value Theory (EVT) provides this unifying mathematical language through the Generalized Extreme Value (GEV) distribution, which encompasses three fundamental classes: Weibull for bounded systems, Gumbel for exponentially-decaying processes, and Fr\'echet for heavy-tailed phenomena \cite{beirlant2006,coles2001introduction}. The shape parameter $\xi$ of the GEV distribution serves as a diagnostic tool, revealing whether extreme events arise from linear dynamics ($\xi \rightarrow 0$), or from nonlinear mechanisms such as modulation instability ($\xi > 0$) that can generate power-law tails \cite{onorato2013}. This framework has proven remarkably successful in connecting microscopic dynamics to macroscopic extreme events across disciplines \cite{gomes2015}.

A particularly powerful extension of classical extreme value theory is the Extreme Value Birnbaum-Saunders (EVBS) distribution, which synergistically combines the physical motivation of fatigue-life models with the mathematical generality of GEV statistics \cite{ferreira2012,leiva2016extreme}. Originally derived from cumulative damage processes in materials science \cite{birnbaum1969}, the EVBS distribution naturally captures the asymmetric, positive-valued nature of many physical observables while maintaining the flexibility to model both moderate deviations and extreme tails. This dual capability proves essential for phenomena ranging from optical intensity fluctuations to atmospheric pressure variations.

The advent of nonlinear optics has provided unprecedented experimental access to extreme event dynamics in controlled settings. Modulation instability in optical fibers generates breather solutions—localized wave packets that emerge from the nonlinear Schr\"odinger equation and manifest as optical rogue waves \cite{solli2007,hammani2008}. These breathers, including the Peregrine soliton and higher-order rational solutions, represent exact analytical descriptions of extreme events that can be generated on demand \cite{dudley2014}. Similar mechanisms operate in hydrodynamic systems, where the Benjamin-Feir instability drives the formation of oceanic rogue waves through non-resonant wave interactions \cite{janssen2003}. The remarkable correspondence between optical and oceanic extreme events underscores the universal nature of these phenomena \cite{chabchoub2011}.

%Atmospheric and space plasma systems exhibit extreme events driven by turbulent cascades and long-range correlations. The Akasofu $\epsilon$ parameter, measuring energy influx into Earth's magnetosphere, follows Fr\'echet statistics with shape parameter $\gamma > 0$, indicating an algebraic tail that makes extreme space weather events more probable than Gaussian models would predict \cite{Schumann_2012}. This heavy-tailed behavior emerges from fractional L\'evy motion dynamics in the solar wind, where non-Gaussian fluctuations and long-range dependencies fundamentally alter extreme value statistics. 

Atmospheric and space plasma systems, including the solar wind and Earth's magnetosphere, exhibit extreme events influenced by turbulent cascades and long-range correlations \cite{Maloney_2010,Schumann_2012}. The Akasofu $\epsilon$ parameter, a measure of energy influx into Earth's magnetosphere from the solar wind, follows Fréchet statistics with a shape parameter $\gamma > 0$. This algebraic (power-law) tail indicates that extreme space weather events are more probable than predicted by Gaussian models. Although fractional Lévy motion (fLm) is an archetype of a dependent stochastic process known for its heavy-tailed, non-Gaussian fluctuations and long-range dependencies \cite{Maloney_2010, Schumann_2012}, the direct application of a fLm model to the Akasofu $\epsilon$ parameter was found to be not capable of reliably estimating the extreme value statistics of $\epsilon$. This highlights that temporal correlations and nonstationarities in the real data play a crucial role and significantly influence its extreme values, in ways not fully captured by the fLm model. 
Similar statistical signatures, characteristic of extreme events, appear in Titan's atmosphere, where intermittent methane rainstorms are generated through convective instability \cite{hueso2006}.

Biological systems present unique challenges for extreme value analysis due to their inherent complexity and hierarchical organization. DNA replication in \emph{Xenopus laevis} embryos exhibits extreme variability in replication timing, with some genomic regions experiencing delays that follow heavy-tailed distributions \cite{bechhoefer2007}. These extreme events arise from the stochastic nature of replication origin firing combined with the system's need for completion reliability—a biological example of competing constraints generating non-trivial statistics. The multifractal nature of these fluctuations, characterized by non-self-averaging properties, requires extensions of classical extreme value theory that account for strong correlations and scale-dependent behavior \cite{Muzy-1,Muzy-2}.

However, fitting sophisticated models to extreme event data raises a fundamental question in statistical physics: how stable are our inferences? A single anomalous observation—perhaps a particularly severe storm or an experimental outlier—can disproportionately influence parameter estimates and subsequent predictions \cite{cook1986}. This sensitivity is especially acute for heavy-tailed distributions where extreme observations carry substantial statistical weight. Without rigorous methods to assess model stability, the physical interpretation of fitted parameters remains precarious, potentially leading to erroneous conclusions about underlying mechanisms or future risks \cite{nguyen2017}.

The local influence methodology, pioneered by Cook \cite{cook1986} and refined through the conformal normal curvature approach of Poon and Poon \cite{poon1999conformal}, offers a geometric framework for probing model stability. By examining how small perturbations to individual observations affect the likelihood surface, these methods identify influential data points that dominate parameter estimation. This approach proves particularly valuable for extreme value models, where the interplay between bulk behavior and tail characteristics requires careful diagnostic analysis.

This paper addresses these challenges by developing a comprehensive local influence diagnostic toolkit specifically tailored for EVBS regression models applied to extreme events in complex physical systems. We demonstrate that this framework is not merely a statistical refinement but a powerful physical probe capable of revealing how individual extreme events shape our understanding of system dynamics. By applying our methodology to real-world data on anomalous wind gusts—a pressing challenge in modern weather forecasting \cite{lam2023}—we identify specific catastrophic storms as highly influential observations and precisely quantify their impact on model parameters. This analysis provides crucial insights into the stability of extreme event predictions and the robustness of physical interpretations derived from statistical models.
Our approach bridges the gap between abstract extreme value theory and practical applications in physics, offering a rigorous framework for building reliable models of rare events in complex natural systems. By combining the physical realism of EVBS distributions with advanced stability diagnostics, we provide tools that enable physicists to confidently characterize extreme phenomena while acknowledging the inherent uncertainties in finite data. This synthesis of statistical theory and physical insight represents a significant advance in our ability to understand and predict nature's most dramatic events.

The rest of this paper is organized as follows. Section \ref{distribuicoesevbslogevbs} reviews the EVBS and log-EVBS distributions, presenting their definitions and deriving some properties related to the existence of moments and the behavior of the log-EVBS pdf. In Section \ref{mrllogevbs}, we present the EVBS regression model proposed by \cite{leiva2016extreme} and address the joint estimation of its parameters via maximum likelihood. Section \ref{s:sim} contains a simulation study to investigate the asymptotic properties of the MLE under model misspecification. In Section \ref{influencialocal}, we introduce local influence diagnostics for the EVBS regression model, following the framework of \cite{cook1986} and \cite{poon1999conformal}, considering three perturbation schemes: case weighting, response variable perturbation, and explanatory variable perturbation. In Section \ref{aplicacaodadosreais}, we analyze a real dataset relating maximum wind speed to atmospheric pressure, illustrating the potential of the EVBS regression model and the effectiveness of our local influence methodology. Finally, Section \ref{conclusao} concludes with a discussion.

\section{The EVBS and log-EVBS distributions}\label{distribuicoesevbslogevbs}

\subsection{Definitions and some properties}

The GEV distribution is a three-parameter distribution, denoted by ${\rm GEV}({\mu, \sigma,\gamma}),$ where 
$(\gamma, \mu, \sigma) \in \mathbb{R} \times \mathbb{R} \times (0, \infty).$ The location parameter $\mu$ is related to the magnitude; the scale parameter $\sigma$ to the variability; and the shape
parameter $\gamma$ to the heaviness of the tail (it provides information about the probability of occurrence of extreme events \cite{embrechts1999extreme,coles2001introduction,dey2016extreme,gomes2015extreme}. It is widely used
for analysis of hydrologic extremes \cite{katz1992extreme}, \cite{leclerc2007non} and the effect of climate change \cite{cooley2009extreme} \cite{lobeto2021},  among others.

The cumulative distribution function is given by
\begin{equation}
\label{eq:GEV}
G_{\mu, \sigma,\gamma} (x) 
= 
G_{0, 1,\gamma}( (x - \mu)/\sigma ),
\qquad x \in \mathbb{R},
\end{equation}
where
\begin{align*}
G_{0, 1,\gamma}(x) 
%   =
%   \Go_\gamma(z)
=
\begin{cases}
\exp(-e^{-x}) 
& \text{if $\gamma = 0,$} \\
\exp\left\{ -(1+\gamma x)^{-1/\gamma} \right\}  
& \text{if  $1 + \gamma x > 0$ and $\gamma \ne 0,$} \\
0 & \text{if $x \le -1/\gamma$ and $\gamma > 0,$ } \\
1 & \text{if $x \ge -1/\gamma$ and $\gamma < 0.$ }
\end{cases}
\end{align*}
The above parameterization, due to \cite{mises1936distribution} and \cite{jenkinson1955frequency}, generalizes and unifies the Fr\'echet, Gumbel and Weibull trichotomy, in \cite{fisher1928limiting} across all signs of the shape parameter $\gamma$. The support, $S = \{x \in \mathbb{R} : \sigma + \gamma (x - \mu) > 0 \}$ of the GEV distribution depends on $({\mu, \sigma,\gamma})$, and it is especially challenging to develop asymptotic statistical inference based on the MLE.

Given the random variable $X\sim \mbox{GEV}(0,1,\gamma)$, \cite{ferreira2012}
introduced the EVBS distribution with parameters $\alpha >0$, $\beta >0$ and $\gamma\in \mathbb{R}$, denoted by $T\sim\mbox{EVBS}(\alpha,\beta,\gamma)$, through the increasing monotonic transformation of the random variable $X$, given by
\begin{equation}
T=\beta\left(\frac{\alpha X}{2}+ \sqrt{\left(\frac{\alpha X}{2}\right)^2+1}\right)^2.
\label{eq:relgevevbs}
\end{equation}
Therefore, $T$ is a positive random variable and its respective cdf can be written in the form
\begin{equation}
F(t)=\left\{\begin{array}{lcc}
\exp\left(-\left(1+\gamma a_t\right)^{-1/\gamma}\right); & 1+ \gamma a_t>0,  & \mbox{ if }  \gamma \neq 0.\\ 
\exp\left(-\exp\left(-a_t\right)\right); &t>0, &\ \mbox{if}\ \gamma = 0,\\ 
\end{array} \right. 
\label{eq:fdaevbs2}
\end{equation} 
in which $a_t=\frac{1}{\alpha}\left(\sqrt{\frac{t}{\beta}}-\sqrt{\frac{\beta}{t}}\right)$. Furthermore, the probability density function (pdf) of the variable $T$ is given by
\begin{equation}
f(t)=\left\{\begin{array}{lcl}
A_{t} \left(1+ {\gamma}a_t\right)^{-1-1/\gamma} \exp\left(-\left(1+ {\gamma}a_t\right)^{-1/\gamma}\right);& 1+ {\gamma}a_t>0,& \ \mbox{ if } \gamma \neq 0.\\
A_{t}\exp\left(-\exp\left(-a_t\right)-a_t\right); & t>0, &\ \mbox{if}\ \gamma=0. \\
\end{array} \right.
\label{eq:fdpevbs2}
\end{equation}
where $A_{t}=\frac{1}{2\alpha \beta}\left[\left(\frac{t}{\beta}\right)^{-\frac{1}{2}}+\left(\frac{t}{\beta}\right)^{-\frac{3}{2}}\right]$.

Figure \ref{g:fdpevbsalphabeta1gamma06258x8}(a) presents graphs of density functions associated with EVBS distributions with heavy tails and an infinite upper support limit. There is positive asymmetric behavior and strong signs of unimodality. Furthermore, an increase in the flattening of the density curve can be seen as parameter $\alpha$ grows. In Figure \ref{g:fdpevbsalphabeta1gamma06258x8}(b), we observe that the shape parameter $\gamma$ significantly influences the tail behavior and the type of asymmetry of the EVBS distribution. Particularly, when $\gamma=-1.25$, the pdf has a finite upper support limit, and there is evidence that it assumes arbitrarily large values, leading us to conjecture that there is no global maximum inside the support.

\begin{figure}[!htbp]
	\center
	\subfigure[]{\includegraphics[width=0.5\textwidth]{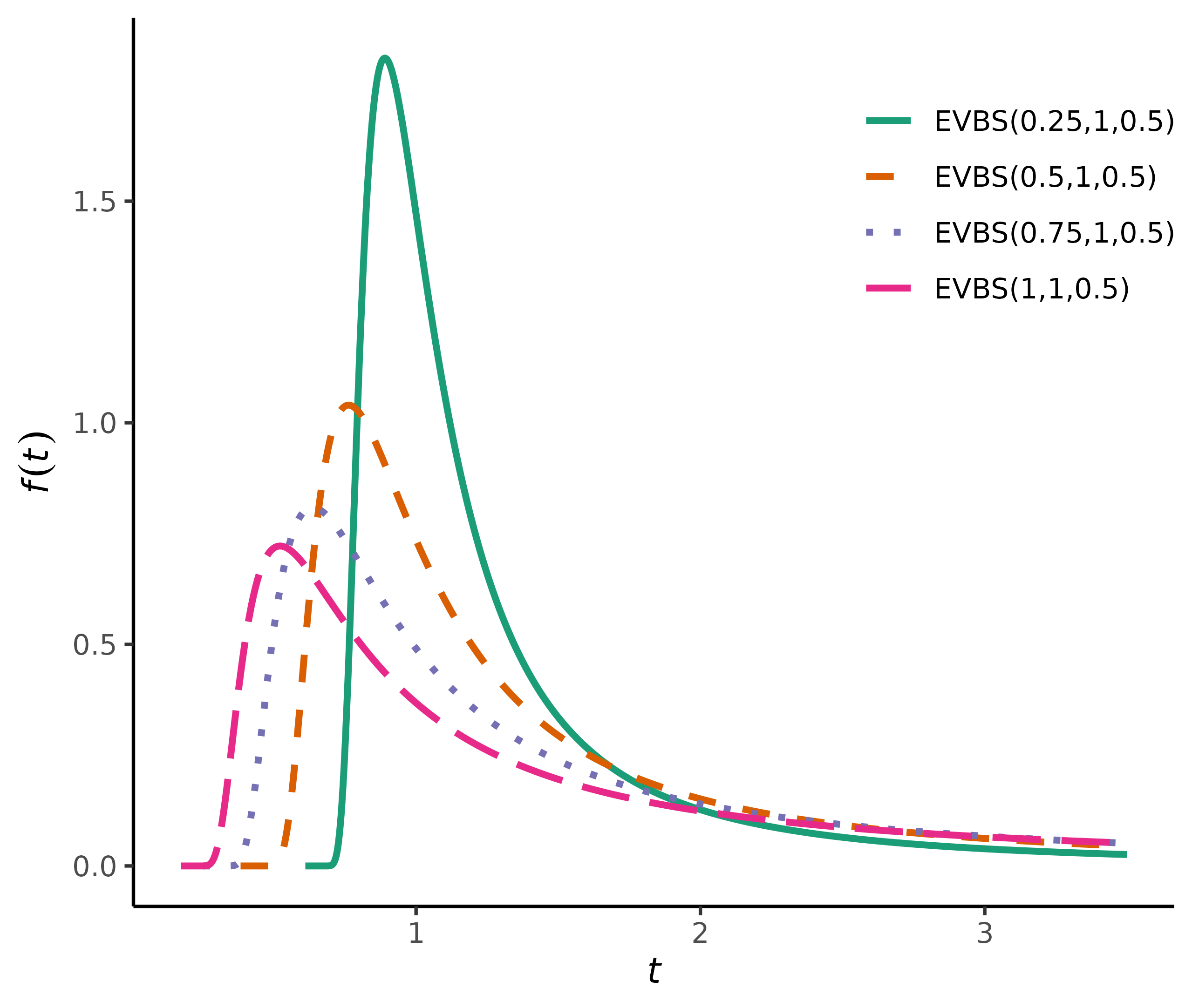} \label{fig1:a}} 
    \hspace{-0.5cm}
	\subfigure[]{\includegraphics[width=0.5\textwidth]{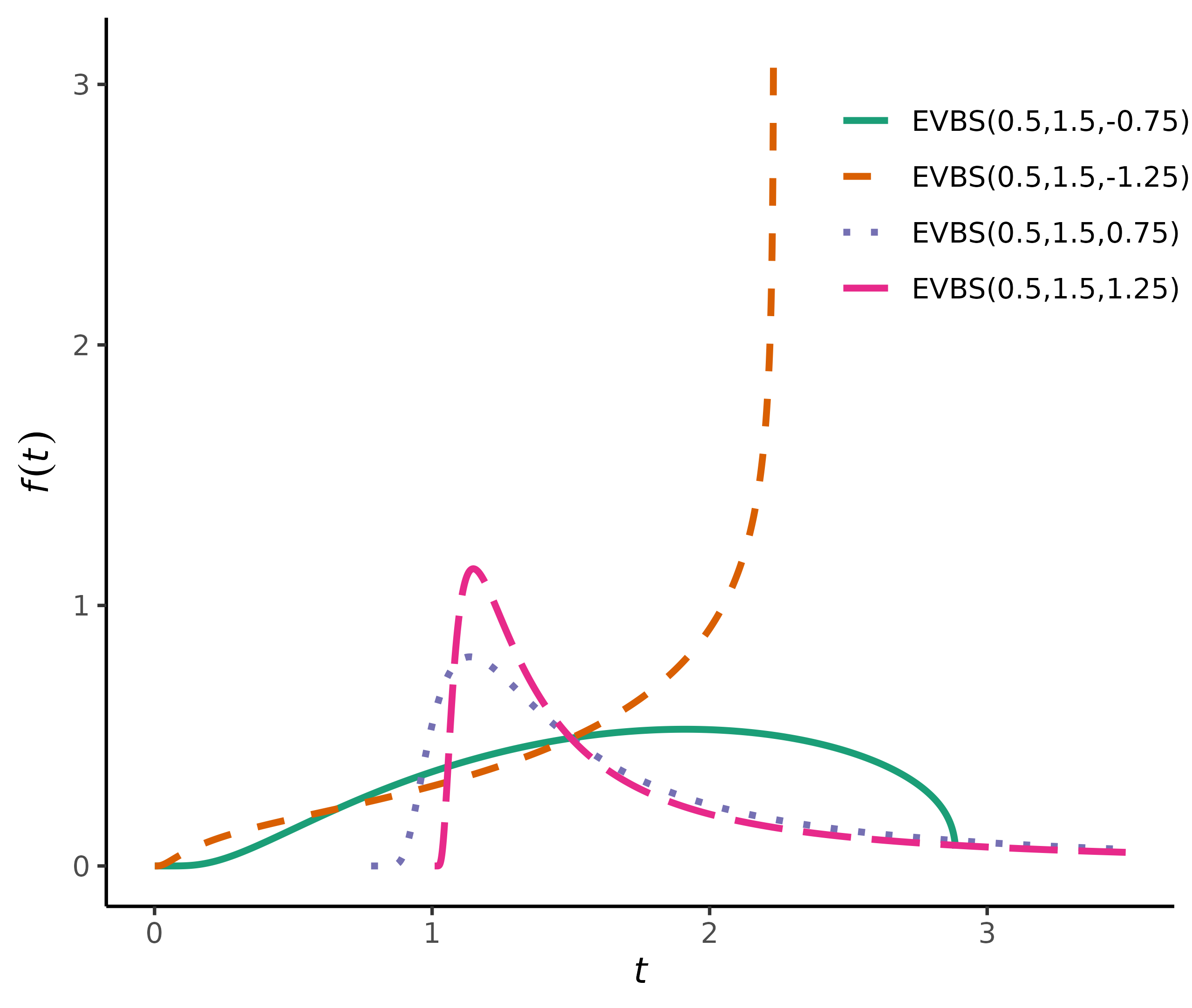} \label{fig1:b}}\\
	\caption{\label{g:fdpevbsalphabeta1gamma06258x8}EVBS pdf graphics for $\alpha \in\{1,0.75,0.50,0.25\}$, $\beta=1$ and $\gamma=0.5$ in (a)  and for $\alpha=0.5$, $\beta=1.5$ and $\gamma\in \{-1.25,-0.75,0.75,1.25\}$  in (b).}
\end{figure}

Among the properties of the EVBS distribution, we highlight that if  $T\sim \mbox{EVBS}(\alpha,\beta,\gamma)$ and $c>0$, then $cT\sim \mbox{EVBS}(\alpha,c\beta,\gamma)$. \cite{ferreira2012}  derived other properties of an EVBS distribution, they have obtained a characterization of the max-domain of attraction of the variable  $T$ in terms of the domain of $X$.

Next, we introduce the log-EVBS distribution, relating it directly to the GEV distribution. If $X\sim \mbox{GEV}(0,1,\gamma)$ with $\gamma \in \mathbb{R}$, for any  $\alpha>0$ and $\eta \in \mathbb{R}$ the variable $Y$ generated through the transformation $Y=\eta +2\mathrm{arcsinh}\left(\frac{\alpha X}{2}\right)$ is such that $Y=\log(T)$, where $T\sim \mbox{EVBS}(\alpha, e^\eta,\gamma)$. We have observed that the cdf, $F_Y$, of the variable $Y= \eta +2\mathrm{arcsinh}\left(\frac{\alpha X}{2}\right)$ is related to the function $G_{\gamma}$ through the expression
\[
F_{Y}(y)=G_{\gamma}\left(\frac{2}{\alpha}\mathrm{sinh}\left(\frac{y-\eta}{2}\right)\right), 
\]
so that $1+\frac{2\gamma}{\alpha}\mathrm{sinh}\left(\frac{y-\eta}{2}\right)>0$. Consequently, the pdf of $Y$ is given by
\begin{equation}
f_{Y}(y)= \frac{1}{2}\xi_1(y) (1+\gamma\xi_2(y))^{-1-\frac{1}{\gamma}} \exp(-(1+\gamma \xi_2(y))^{-\frac{1}{\gamma}})
\label{f:fdplogevbs},
\end{equation}
where $\xi_{1}(y)=\frac{2}{\alpha}\cosh\left(\frac{y-\eta}{2}\right)$, $\xi_{2}(y)=\frac{2}{\alpha}\mathrm{sinh}\left(\frac{y-\eta}{2}\right)$ and $1+\gamma\xi_{2}(y)>0$. The function expressed in equation (\ref{f:fdplogevbs}) is equivalent to the pdf of the variable $\log(T)$, with $T\sim \mbox{EVBS}(\alpha,e^\eta,\gamma)$, obtained by \cite{leiva2016extreme}. Therefore, the variable $Y= \eta +2\mathrm{\rm arcsinh}(\frac{\alpha X}{2})\sim \mbox{log-EVBS}(\alpha,\eta,\gamma)$. The justification of the case where $\gamma=0$  is analogous.

Given a random variable $Y\sim \mbox{log-EVBS}(\alpha,\eta,\gamma)$, when considering  $\xi_{1}(y)=\frac{2}{\alpha}\cosh\left(\frac{y-\eta}{2}\right)$ and $\xi_{2}(y)=\frac{2}{\alpha}\mathrm{sinh}\left(\frac{y-\eta}{2}\right)$, the pdf of $Y$ admits being written in the form
\begin{equation}
f(y)=\left\{\begin{array}{lll}
\frac{1}{2}\xi_1(y) (1+\gamma\xi_2(y))^{-1-\frac{1}{\gamma}} \exp(-(1+\gamma \xi_2(y))^{-\frac{1}{\gamma}}); & 1+\gamma \xi_2(y)>0,& \mbox{if } \gamma \neq 0.\\
\frac{1}{2} \xi_1(y) \exp(-\xi_2(y)-\exp(-\xi_2(y)));& y\in \mathbb{R}, & \mbox{if } \gamma=0. \\
\end{array} \right.
\label{f:fdplogaevbsxi}
\end{equation}

In developing the regression model associated with the EVBS distribution, \cite{leiva2016extreme} have admitted a log-EVBS distribution for the stochastic component of the model. Therefore, the identification of pdf properties given the log-EVBS variable, concerning monotonicity and the existence of maximum locations, will add information to the error behavior of the model, and cooperate for the implementation of inferential procedures. Let us consider a random variable  $Y\sim \mbox{log-EVBS}(\alpha,\eta,\gamma)$ whose pdf is given by equation (\ref{f:fdplogaevbsxi}), the following statements are valid:
\begin{proposition}
	\label{prop:1.1.1}
	If $\gamma<-1$, then the function  $f$ is strictly increasing in the open interval $(\eta,\eta+2{\mathrm{arcsinh}}\left(-{\alpha}/{2\gamma}\right))$.
\end{proposition}

\begin{proposition}
	\label{prop:1.1.2}
	In the case where  $\gamma=-1$. If $\alpha\geq 4$, then the function $f$ is increasing in the intervals $(-\infty,\eta+2\mathrm{arcsinh}\left(-\frac{\alpha}{4}-\frac{1}{4}\sqrt{\alpha^2-16}\right))$ and  $(\eta+2\mathrm{arcsinh}\left(-\frac{\alpha}{4}+\frac{1}{4}\sqrt{\alpha^2-16}\right),\eta+2{\mathrm{arcsinh}}\left(-{\alpha}/{2\gamma}\right))$.On the other hand,  $f$ is decreasing in the interval $$\left(\eta+2\mathrm{arcsinh}\left(-\frac{\alpha}{4}-\frac{1}{4}\sqrt{\alpha^2-16}\right),\eta+2\mathrm{arcsinh}\left(-\frac{\alpha}{4}+\frac{1}{4}\sqrt{\alpha^2-16}\right)\right).$$ Also, if $0<\alpha<4$, then $f$ is strictly increasing in $(-\infty, \eta+2 \mathrm{arcsinh}({\alpha}/{2}))$.
\end{proposition}

\begin{proposition}
	\label{prop:1.1.3}
	If $\gamma=0$ and $0<\alpha<2$, then $\eta$ is a local maximum of $f$. Furthermore, $f$ is strictly increasing in the interval  $(-\infty,\eta)$. On the other hand, if $\alpha>2$, then $\eta$ is a local minimum of  $f.$
\end{proposition}

\begin{proof}
	The proofs of the propositions \ref{prop:1.1.1}--\ref{prop:1.1.3} above are given in Appendix \ref{pr1}--\ref{pr3}.
\end{proof}

Next, we will present some illustrations that describe the behavior of a pdf from the log-EVBS distribution family. First, in Figure \ref{logevbsalphaeta0gamma0}(a) there is evidence of unimodality when $\gamma=0$ and $\alpha$ approaches zero, with a trend towards symmetry; but there is clear evidence of bimodality when $\alpha>2$. Furthermore, there is evidence of an increase in dispersion and the flattening of the curve, as parameter $\alpha$ grows. In Figure  \ref{logevbsalphaeta0gamma0}(b),  we can observe, once again, the expressive effect of the parameter $\gamma$  on the behavior of the tail and the type of asymmetry of the distribution.  When $\gamma=-1.05$ and $\alpha=1$, the monotonic increasing behavior of the density curve in its support stands out, with indications that the pdf assumes arbitrarily large values near the upper limit of the support.

\begin{figure}[!htbp]
	\center
	\subfigure[]{\includegraphics[width=0.5\textwidth]{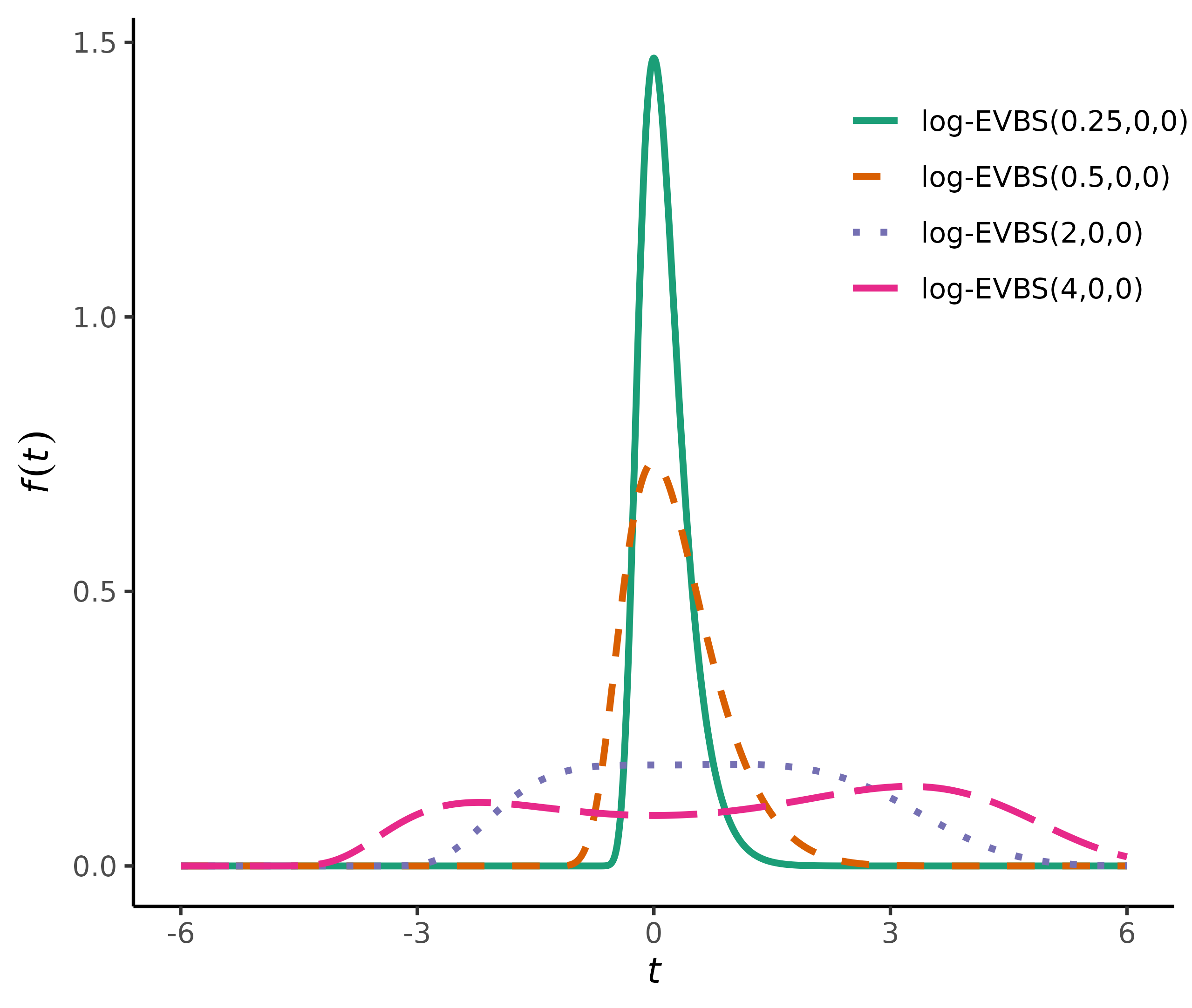} \label{fig2:a}} 
    \hspace{-0.5cm}
	\subfigure[]{\includegraphics[width=0.5\textwidth]{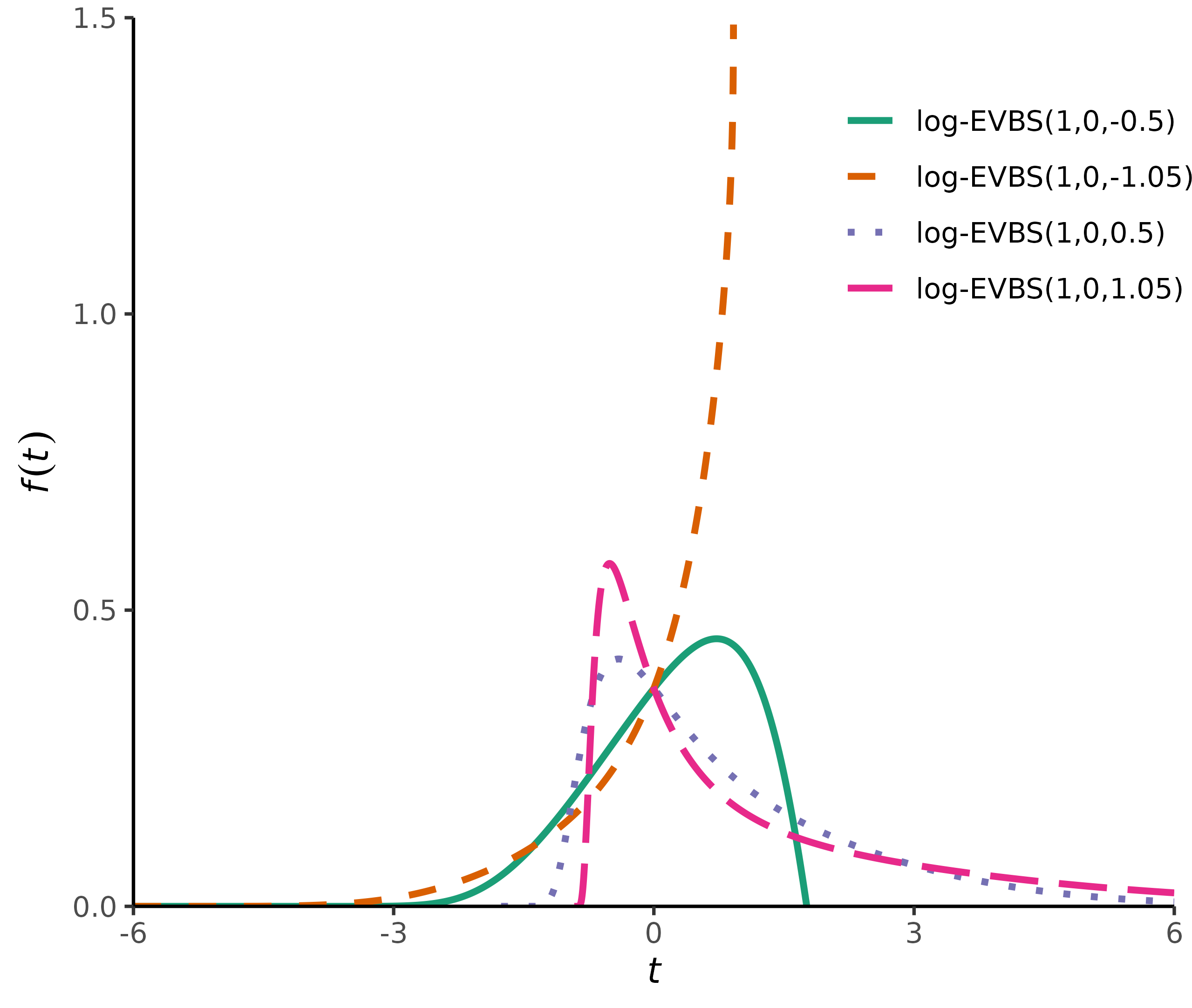} \label{fig2:b}}\\
	\caption{\label{logevbsalphaeta0gamma0}Graphics of the pdf of a log-EVBS for $\alpha \in \{0.25,0.50,2,4\}$, $\eta=0$,
		and $\gamma=0$ (a) and for $\alpha=1$, $\eta=0$ and	 $\gamma\in \{-1.05,-0.5,0.5,1.05\}$ (b).}
\end{figure}

\subsection{Moments}

In this subsection, we will deduce a necessary and sufficient condition for the existence of the first two central moments of the EVBS distribution. This is relevant because the finiteness of moments is a key hypothesis in limit theorems that underpin much of statistical mechanics.
We will resort to the relation given by equation (\ref{eq:relgevevbs}) and use some results about the existence of moments of a variable that follows a GEV distribution. 
\begin{proposition}
	\label{prop:1.2.1}
	Let $T$ be a random variable such that  $T\sim\mbox{EVBS}(\alpha,\beta,\gamma)$. Then, the following affirmations are valid:
	\begin{equation}
	\mbox{E}(T)< \infty, \mbox{ if and only if } \gamma<\frac{1}{2}.
	\end{equation} 
	
	\begin{equation}
	\mbox{E}(T^2)< \infty, \mbox{ if and only if } \gamma<\frac{1}{4}.
	\end{equation}
\end{proposition}

\begin{proof} 
	By developing the square present in equation (\ref{eq:relgevevbs}) and using some properties of the modular function, combining it with the fact that the square root function is strictly increasing, we have established the following simultaneous inequality:
	\[
	T\leq \frac{\beta}{2}\left(\alpha^2 X^2+2+|\alpha X|\sqrt{\alpha^2X^2+4}\right)< \frac{\beta}{2}\left(\alpha^2 X^2+2+\sqrt{(\alpha^2X^2)^2+4\alpha^2X^2+4}\right).
	\]
	Therefore,
	\begin{equation}
	T\leq \beta(\alpha^2X^2+2).
	\label{d:limsupEVBS}
	\end{equation}
	We know that if $X\sim \mbox{GEV}(0,1,\gamma)$ and $\gamma<\frac{1}{2}$, then $\mbox{E}(X^2)<\infty$. Consequently, using inequality (\ref{d:limsupEVBS}), we will have that if $\gamma < \frac{1}{2}$, then $\mbox{E}(T)$ is finite. It follows from inequality (\ref{d:limsupEVBS}) that $T^2\leq \beta^2(\alpha^4X^4+4\alpha^2X^2+4)$. Hence, if $X\sim \mbox{GEV}(0,1,\gamma)$ and $\gamma< \frac{1}{4}$, then $\mbox{E}(X^4)< \infty$ and, consequently, $\mbox{E}(T^2)$ is finite.
	
	On the other hand, in order to show that $\mbox{E}(T)<\infty$ implies $\gamma<\frac{1}{2}$,  we will prove the contrapositive $\gamma\geq 1/2$ implies $\mbox{E}(T)=\infty$. From the inequality
	\begin{equation}
	\frac{\alpha X}{2}+{\sqrt{ \left(\frac{\alpha X}{2}\right)^2 + 1 } } > \frac{\alpha X}{2}+ \left | \frac{\alpha X}{2}\right| \geq 0,
	\end{equation}
	we obtain 
	\begin{equation}
	T=\beta\left(\frac{\alpha X}{2}+ \sqrt{\left(\frac{\alpha X}{2}\right)^2+1}\right)^2 > \beta\frac{\alpha^2}{2}(X^2+X|X|).
	\label{d:limiteinfT}
	\end{equation}
	At this point, we observe that the following result is valid: if $Y$ is a random variable in a probability space $(\Omega, {\cal  B}(\R),P)$ and  $Y^2$ is not integrable, then $Y^2+Y|Y|$ is not integrable. Furthermore, let us consider the fact that a random variable, whose distribution is $\mbox{GEV}(0,1,\gamma)$ has a finite second moment if, and only if, the parameter $\gamma < \frac{1}{2}$. Hence, if $\gamma\geq \frac{1}{2}$, then $E(X^2)=\infty$ and, consequently, using inequality (\ref{d:limiteinfT}), we have $E(T)=\infty$. Similar reasoning allows us to establish the inequality $$T^2> \frac{1}{2}\alpha^4\beta^2(X^4+X^3|X|).$$ Hence, if $\gamma\geq 1/4$, then $T^2$ is not integrable. 
\end{proof}

With the result expressed in Proposition \ref{prop:1.2.1}, we can explain sufficient conditions for the existence of the moments of the variable $Y=\log T$ with log-EVBS distribution. Note that the function $\varphi:(0,\infty)\to \R, \ \varphi(x)=-\log x$ is convex, so by Jensen's Inequality we have that if $Y=\log(T)\sim \mbox{log-EVBS}(\alpha,\eta,\gamma)$ and $\gamma<\frac{1}{2}$, then $\mbox{E}(Y)$ is finite. On the other hand, given that $x>\log(x)$ for all $x\in (0,+\infty)$, if $T>1$, then $T>\log(T)>0$. Hence, $(\log(T))^2<T^2$. Thus, if  $\gamma<\frac{1}{4}$, then $\mbox{E}(Y^2)< \infty$.

\section{Extreme Value Birnbaum-Saunders regression model}\label{mrllogevbs}

Let us consider the EVBS regression model
\begin{equation}
T_i=\varrho_i\delta_i=\exp(\mathbf{x}_{i}^{\top}{\mbox{\boldmath{${\beta}$}}})\delta_i,\quad i=1,\ldots,n, 
\label{mrllevbs}
\end{equation} where $T_i$ is the observed response variable, $\mathbf{x}_{i}^{\top}=(1,x_{i1},\ldots,x_{i(p-1)})$ are the values of the control parameters (or covariates), ${\mbox{\boldmath{${\beta}$}}}=(\beta_0,\beta_1,\ldots,\beta_{p-1})^{\top}$ is a vector of unknown parameters, and $\delta_i\sim \mbox{EVBS}(\alpha,1,\gamma)$ is the stochastic (or noise) component of the model.
Applying the logarithm function to the variable  $T_i$ given in (\ref{mrllevbs}), we obtain the following log-linear regression model:
\begin{equation}
Y_i=\log(T_i)=\mathbf{x}_{i}^{\top}{\mbox{\boldmath{${\beta}$}}}+\varepsilon_i, \quad i=1,...,n,
\label{mrlevbs}
\end{equation} in which  $Y_i\sim \mbox{log-EVBS}(\alpha,\eta_i,\gamma)$, where $\eta_i=\mathbf{x}_{i}^{\top}{\mbox{\boldmath{${\beta}$}}}$.  Note that if $\gamma=0$, the support of $Y$ is the set of real numbers; but, if $\gamma>0$, then $y_i> \eta_i +2{\mathrm{arcsinh}}\left(-{\alpha}/{2\gamma}\right)$; on the other hand, if $\gamma<0$, then $y_i< \eta_i +2{\mathrm{arcsinh}}\left(-{\alpha}/{2\gamma}\right)$.

\subsection{Likelihood inference}

In this subsection, we have derived expressions for the score function and the Hessian matrix of the log-likelihood function of the parameter vector ${\mbox{\boldmath $\theta$}}=({\mbox{\boldmath $\beta$}}^{\top}, \alpha, \gamma)^{\top}$ associated with the model given in  (\ref{mrlevbs}), equivalent to those obtained by \cite{leiva2016extreme}, which enable the determination of maximum likelihood estimates through non-linear optimization algorithms.

Given a $ \mathbf{y}=(y_1,y_2,\ldots,y_n)^{\top}$ sample of independent observations associated with the model (\ref{mrlevbs}),  the log-likelihood function of ${\mbox{\boldmath $\theta$}}=({\mbox{\boldmath $\beta$}}^{\top}, \alpha, \gamma)^{\top}$ can be expressed as
\begin{equation}
\ell({\mbox{\boldmath $\theta$}};\mathbf{y})=\left\{\begin{array}{ll}
-n\log 2 +  \displaystyle \sum_{i=1}^{n}\log(\xi_{i1})
-\left(1+\frac{1}{\gamma}\right) \displaystyle  \displaystyle \sum_{i=1}^{n}\log(1+\gamma\xi_{i2})- \displaystyle \sum_{i=1}^{n}(1+\gamma\xi_{i2})^{-\frac{1}{\gamma}}; &  \gamma \neq 0.\\
n\log 2 +  \displaystyle \sum_{i=1}^{n}\log(\xi_{i1})- \displaystyle \sum_{i=1}^{n}\xi_{i2}- \displaystyle \sum_{i=1}^{n}\exp(-\xi_{i2}); & \gamma=0. \\
\end{array} \right.
\label{eqmvlogevbs}
\end{equation}
where $\xi_{i1}=\frac{2}{\alpha}\mathrm{cosh}\left(\frac{y_i-\mathbf{x}_{i}^{\top}\bbeta}{2}\right)$ and  $\xi_{i2}=\frac{2}{\alpha}\mathrm{sinh}\left(\frac{y_i-\mathbf{x}_i^{\top}\bbeta}{2}\right)$, for each $i=1,2,\ldots,n$.

The score function is given by: 

\begin{equation}
\U(\ttheta,\mathbf{y})=\frac{\partial \ell(\ttheta, \mathbf{y})}{\partial \ttheta}=
\left(
\begin{array}{l}
\frac{\partial} {\partial \bbeta} \ell(\ttheta,\mathbf{y})\\
\frac{\partial} {\partial \alpha} \ell(\ttheta,\mathbf{y})\\
\frac{\partial} {\partial \gamma} \ell(\ttheta,\mathbf{y})\\
\end{array}
\right)_{(p+2)\times 1}, \mbox{ for }\gamma\ne 0 
\label{fescoredifzero}
\end{equation}
and

\begin{equation}
\U(\ttheta,\mathbf{y})=
\left(
\begin{array}{l}
\frac{\partial} {\partial \bbeta} \ell(\ttheta,\mathbf{y})\\
\frac{\partial} {\partial \alpha} \ell(\ttheta,\mathbf{y})\\
\end{array}
\right)_{(p+1)\times 1}, \mbox{ for } \gamma=0,
\label{fescoreigualzero}
\end{equation}
whose entries are given by: 

$\gamma\neq 0:$
\begin{equation}
\left\{\begin{array}{cc}
\frac{\partial}{\partial \beta_j}\ell({\mbox{\boldmath $\theta$}};\mathbf{y})=\frac{1}{2}\displaystyle\sum_{i=1}^{n}x_{ij} \left[ \frac{\xi_{i1}}{1+\gamma \xi_{i2}}\left(1+\gamma - (1+\gamma \xi_{i2})^{-1/\gamma}\right) -\frac{\xi_{i2}}{\xi_{i1}}\right], &\quad j=0,1,\ldots,p-1. \\
\frac{\partial}{\partial \alpha}\ell({\mbox{\boldmath $\theta$}};\mathbf{y}) =-\frac{n}{\alpha} + \frac{1}{\alpha} \displaystyle \sum_{i=1}^{n}\frac{\xi_{i2}}{1+\gamma \xi_{i2}}\left(1+\gamma-(1+\gamma\xi_{i2})^{-1/\gamma}\right). & {} \\

\frac{\partial}{\partial \gamma}\ell({\mbox{\boldmath $\theta$}};\mathbf{y}) =-\frac{1}{\gamma} \displaystyle \sum_{i=1}^{n}\frac{\xi_{i2}}{1+\gamma \xi_{i2}}\left[1+\gamma -(1+\gamma\xi_{i2})^{-1/\gamma} \right]+ & {} \\
\hspace{2.0cm} + \frac{1}{\gamma^2}\displaystyle \sum_{i=1}^{n} \log(1+\gamma \xi_{i2})\left[1-(1+\gamma \xi_{i2})^{-1/\gamma}\right].
\end{array} \right.
\label{eq:semvlogbs1}
\end{equation}
and $\gamma= 0:$
\begin{equation}
\left\{\begin{array}{lc}
\frac{\partial}{\partial \beta_j}\ell({\mbox{\boldmath $\theta$}};\mathbf{y})=\frac{1}{2}\displaystyle\sum_{i=1}^{n}x_{ij}\left(\xi_{i1}(1-\exp(-\xi_{i2}))-\frac{\xi_{i2}}{\xi_{i1}}\right), &\mbox{ for }\quad j=0,1,\ldots,p-1. \\
\frac{\partial}{\partial \alpha}\ell({\mbox{\boldmath $\theta$}};\mathbf{y}) =-\frac{n}{\alpha} + \frac{1}{\alpha} \displaystyle \sum_{i=1}^{n}\xi_{i2}(1-\exp(-\xi_{i2})). & {} \\
\end{array} \right.
\label{eq:semvlogbs}
\end{equation}

The MLE of $\ttheta$, denoted by $\hatttheta=(\hatbbeta,\hat{\alpha},\hat{\gamma})^{\top}$, must be such that $\U(\hatttheta)=\mathbf{0}$. However, in the case of the log-EVBS model, the system of equations $\U(\ttheta)=\mathbf{0}$ does not present an analytic solution, requiring non-linear numerical optimization procedures in order to obtain solutions. Furthermore, the Hessian matrix of the log-likelihood function evaluated in  $\hatttheta$, denoted by $\ddL_{\hatttheta}$, is given by 
$$\ddL_{\hatttheta}=\frac{\partial }{\partial \ttheta}\U(\ttheta)\left|_{\ttheta=\hatttheta} \right.={\frac{\partial^2 }{\partial \ttheta \partial \ttheta^{\top}}\ell(\ttheta)}\left|_{\ttheta=\hatttheta}, \right.$$ whose terms we present in Appendix \ref{apendiceB}. For $\hatttheta$ to be a local maximum of  $\ell(\cdot)$, it is sufficient if the conditions $\U(\hatttheta)=\mathbf{0}$ and $\ddL_{\hatttheta}$ are defined as negative. It is known that the matrix $-\ddL_{\hatttheta}^{-1}$, called the inverse of the observed Fisher information, plays a relevant role in obtaining numerical estimates and in the calculation of curvatures. 

At this point, we will make a consideration. Given the relationship between the EVBS and GEV distributions, the maximum likelihood estimation of the parameter vector of the log-EVBS regression model will present similar difficulties to those found in the GEV family, especially the ones arising from the fact that the support of this distribution depends on parameter values, violating one of the classic regularity conditions. Therefore, the verification of hypotheses to ensure asymptotic properties of the MLE is a problem that deserves to be investigated, but it is far from the objective of this work. \cite{Dombry2015}, \cite{Bucher2017}, and  \cite {DombryFerreira2019}  present sufficient conditions to ensure the consistency and asymptotic normality of the MLE in the GEV family. In this work we will maintain these assumptions, assuming that the monotonic transformation on the GEV variable will not relax them. Furthermore, we will consider the assumption that the parameter $\alpha\leq 2$, admitted by \cite{Rieck1991}, to avoid multiplicity of solutions in the system of maximum likelihood equations. 

\section{Simulations}
\label{s:sim}

In this section, we have conducted a Monte Carlo simulation study to evaluate the performance of the estimators under some scenarios. It is noteworthy that this study is necessary, as we have performed the estimation of the parameter vector $\ttheta=(\bbeta,\alpha,\gamma)^{\top}$ jointly, unlike the study conducted by \cite{leiva2016extreme}. 

We have considered $5,000$ replicas for different sample sizes, $n\in\{60,120,180\}$. We have assumed the regression model in which the response $Y_i=\eta_i+\varepsilon_{i}$, $i=1,...,n$,  where $Y_i=\log(T_i)$ , $\eta_i=\beta_0+\beta_1x_i$, and $\varepsilon_i\sim \mbox{log-EVBS}(\alpha,0,\gamma)$.  In addition, we have explored three distinct scenarios, whose characterizations will be given below, and obtained empirical approximations for the performance of the maximum likelihood estimators of the model parameters: the absolute bias, the root mean square error (RMSE), and the coverage probability (CP) for the confidence intervals of the parameters with a nominal level fixed at $95\%$. This is under the hypothesis of asymptotic normality of the MLE. Finally, we emphasize that, to numerically obtain the maximum point of the log-likelihood function, we have used the {\tt optim()} function, contained in the statistical software {\sf R}, combined with the BFGS method, and the use of an analytic Hessian matrix. 

\subsection*{Characterization of the scenarios.}

\begin{enumerate}
	\item In Scenario 1, we fixed the model parameters at $\beta_0=0.5$, $\beta_1=0.5$, $\alpha=0.5$, and $\gamma \in\{-0.20,0,0.20\}$. Furthermore, we have assumed that the values of the explanatory variable $X$ belong to the unit interval $(0,1)$. The purpose is to verify if there is empirical support that corroborates the hypothesis of consistency of the MLE of the model parameters.  
	\item Scenario 2 is characterized by leverage imposed on the values of the explanatory variable $X$, so that $10\%$ of its values are randomly fixed in the interval (5,10). Again, the parameters are fixed at $\beta_0=0.5$, $\beta_1=0.5$, $\alpha=0.5$, and $\gamma \in\{-0.20, 0, 0.20\}$. The objective is to investigate the behavior of the estimates in situations in which extreme values for the response are induced by leverage on the values of the explanatory variable.  
	\item  In Scenario 3, we have imposed changes in the stochastic component of the model, so that the response variable aggregates values from EVBS distributions, in which different values for parameter $\alpha$ are considered. The true value of parameter $\alpha$ is $1.2$. However, we have ensured that $10\%$ of the sample values were associated with $\alpha=0.7$. The objective is to verify if the maximum likelihood estimates are sensitive to this mixture of distributions. We set the other parameters at  $\beta_0=1$, $\beta_1=2$, and  $\gamma\in \{-0.20, 0, 0.20\}$.
\end{enumerate}

As we can see in Tables \ref{tabela1cenario1} and \ref{tabela2cenario1}, corresponding to Scenario 1, the biases of the MLEs and their RMSEs tend to zero as the sample size grows. For instance, when $\gamma=0.20$ and the sample size increases through 60, 120, and 180, the bias of $\hat{\gamma}$ decreases towards zero, with values of $0.009$, $0.004$, and $0.003$, respectively. In the same situation, the RMSE also decreases. A similar behavior occurs for the estimates of the other model parameters. Hence, the results empirically support the consistency of the MLE, $\hatttheta$. The results in Tables \ref{tabela1cenario2} and \ref{tabela2cenario2}, corresponding to Scenario 2, suggest that the presence of extreme observations in the response, induced by leverage in the explanatory variable, are satisfactorily captured by the model, not significantly affecting the MLE bias. The RMSE for the $\beta_1$ coefficient estimator shows a notable decrease in Scenario 2 compared to Scenario 1. In Scenario 1 (without leverage), the RMSE values were $0.201$, $0.134$, and $0.111$ for sample sizes of 60, 120, and 180, respectively (Table \ref{tabela2cenario1}). In contrast, in Scenario 2 (with leverage), these values dropped to $0.031$, $0.022$, and $0.018$ (Table \ref{tabela2cenario2}).

The results for Scenario 3, shown in Tables \ref{tabela1cenario3} and \ref{tabela2cenario3}, indicate a substantial relative increase in the bias of the MLE for parameter $\alpha$ compared to the previous scenarios. Furthermore, there is an increase in the RMSE of the MLE of $\alpha$ compared to the values in the tables for Scenarios 1 and 2. Finally, the results in Tables \ref{tabela2cenario1}, \ref{tabela2cenario2}, and \ref{tabela2cenario3}, which show the empirical coverage probabilities for the confidence interval of parameter $\alpha$, reveal that the mixture of distributions caused by the change in $\alpha$ was detected by the MLE. The coverage rate for the interval of $\alpha$, shown in the penultimate column of Table \ref{tabela2cenario3}, was slightly affected compared to those in Tables \ref{tabela2cenario1} and \ref{tabela2cenario2}, which remained close to the nominal theoretical level of $95\%$. 

\begin{table}[!htbp]
	% \scalefont{0.75}
	\begin{center}
		\caption{Empirical mean and absolute bias of the MLEs under Scenario 1. } 
                \vspace{0.2cm}
		\label{tabela1cenario1}
		\begin{tabular}{@{}c|c|cccc|c|cccc @{}}
			\hline
			\multicolumn{6}{@{}c|}{MLE}& 
			\multicolumn{5}{@{}c}{BIAS} \\ \hline
			$n$ & $\gamma$ & $\hat{\beta}_0$ & $\hat{\beta}_1$ & $\hat{\alpha}$ & $\hat{\gamma}$ & $\gamma$ & $\hat{\beta}_0$ & $\hat{\beta}_1$ & $\hat{\alpha}$ & $\hat{\gamma}$ \\ \hline
			{}    &-0.20&0.508&0.505&0.488&-0.220          &-0.20&0.008&0.005&-0.012&-0.020\\
			{60}  &0    &0.508&0.501&0.485&-0.004          &0    &0.008&0.001&-0.015&-0.004 \\
			{}    &0.20 &0.507&0.497&0.483&0.209           &0.20&0.007&-0.003&-0.017&0.009\\ \hline 
			
			{}     &-0.20&0.506&0.500&0.495&-0.211        &-0.20&0.006&0.000&-0.005&-0.011 \\
			{120}  &0   &0.504&0.500&0.493&-0.002         &0    &0.004&0.000&-0.007&-0.002 \\
			{}     &0.20 &0.503&0.499&0.492&0.204    	  &0.20 &0.003&-0.001&-0.008&0.004\\ \hline
			
			{}    &-0.20&0.504&0.500&0.496&-0.207        &-0.20&0.004&0.000&-0.004&-0.007 \\
			{180}  &0    &0.502&0.501&0.495&-0.001       &0   &0.002&0.001&-0.005&-0.001\\
			{}    &0.20  &0.501&0.501&0.495&0.203        &0.20  &0.001&0.001&-0.005&0.003 \\ \hline
		\end{tabular}
	\end{center}
\end{table}

\begin{table}[!htbp]
	% \scalefont{0.75}
	\begin{center}
		\caption{Empirical root mean square error (RMSE) and coverage probability (CP) of the MLEs under Scenario 1.}
        \vspace{0.2cm}
		\label{tabela2cenario1}
		\begin{tabular}{@{}c|c|cccc|c|cccc@{}}
			\hline
			\multicolumn{6}{@{}c|}{RMSE}& 
			\multicolumn{5}{@{}c}{CP} \\ \hline
			$n$ & $\gamma$ & $\hat{\beta}_0$ & $\hat{\beta}_1$ & $\hat{\alpha}$ & $\hat{\gamma}$ & $\gamma$ & $\hat{\beta}_0$ & $\hat{\beta}_1$ & $\hat{\alpha}$ & $\hat{\gamma}$ \\ \hline
			{}     &-0.20&0.152&0.225&0.054&0.112         &-0.20&0.927&0.922&0.918&0.924 \\
			{60}  &0    &0.151&0.221&0.056&0.116          &0    &0.926&0.926&0.913&0.933 \\
			&0.20&0.141&0.201&0.062&0.129          &0.20 &0.923&0.920&0.909&0.938 \\
			\hline 
			
			{}    &-0.20&0.093&0.150&0.036&0.067        &-0.20&0.944&0.943&0.941&0.935\\
			{120} &0   &0.092&0.148&0.038&0.073         &0    &0.947&0.944&0.933&0.939 \\
			{}    &0.20 &0.086&0.134&0.042&0.083        &0.20 &0.946&0.940&0.934&0.940 \\
			\hline 
			
			{}    &-0.20&0.074&0.124&0.029&0.052         &-0.20&0.946&0.947&0.947&0.942 \\
			{180} &0   &0.074&0.123&0.030&0.058         &0    &0.945&0.944&0.940&0.944\\
			{}    &0.20 &0.069&0.111&0.033&0.066        &0.20 &0.944&0.940&0.939&0.943 \\ 
			\hline
		\end{tabular}
	\end{center}
\end{table}

\begin{table}[!htbp]
	% \scalefont{0.75}
	\begin{center}
		\caption{Empirical mean and absolute bias of the MLEs under Scenario 2.}
                \vspace{0.2cm}
		\label{tabela1cenario2}
		\begin{tabular}{@{}c|c|cccc|c|cccc@{}}
			\hline
			\multicolumn{6}{@{}c|}{MLE}& 
			\multicolumn{5}{@{}c}{BIAS} \\ \hline
			$n$ & $\gamma$ & $\hat{\beta}_0$ & $\hat{\beta}_1$ & $\hat{\alpha}$ & $\hat{\gamma}$& $\gamma$ & $\hat{\beta}_0$ & $\hat{\beta}_1$ & $\hat{\alpha}$ & $\hat{\gamma}$ \\ \hline
			{}    &-0.20&0.509&0.502&0.488&-0.219        &-0.20&0.009&0.002&-0.012&-0.019\\
			{60}  &0    &0.500&0.506&0.485&-0.005        &0    &0.000&0.006&-0.015&-0.005\\
			{}    &0.20 &0.494&0.509&0.483&0.207         &0.20&-0.006&0.009&-0.017&0.007\\ \hline 
			
			{}    &-0.20&0.506&0.500&0.494&-0.210         &-0.20&0.006&0.000&-0.006&-0.010 \\
			{120}  &0   &0.501&0.502&0.493&-0.003         &0    &0.001&0.002&-0.007&-0.003 \\
			{}    &0.20 &0.498&0.504&0.492&0.204          &0.20 &-0.002&0.004&-0.008&0.004\\ \hline 
			
			{}    &-0.20&0.503&0.501&0.496&-0.207        &-0.20&0.003&0.001&-0.004&-0.007 \\
			{180}  &0    &0.501&0.502&0.495&-0.002        &0    &0.001&0.002&-0.005&-0.002\\
			{}    &0.20  &0.498&0.503&0.495&0.203         &0.20 &-0.002&0.003&-0.005&0.003 \\ \hline
		\end{tabular}
	\end{center}
\end{table}

\begin{table}[!htbp]
	% \scalefont{0.75}
	\begin{center}
		\caption{Empirical root mean square error (RMSE) and coverage probability (CP) of the MLEs under Scenario 2.}
                \vspace{0.2cm}
		\label{tabela2cenario2}
		\begin{tabular}{@{}c|c|cccc|c|cccc@{}}
			\hline
			\multicolumn{6}{@{}c|}{RMSE}& 
			\multicolumn{5}{@{}c}{CP} \\ \hline
			$n$ & $\gamma$ & $\hat{\beta}_0$ & $\hat{\beta}_1$ & $\hat{\alpha}$ & $\hat{\gamma}$ & $\gamma$ & $\hat{\beta}_0$ & $\hat{\beta}_1$ & $\hat{\alpha}$ & $\hat{\gamma}$ \\ \hline
			{}    &-0.20&0.085&0.032&0.054&0.109         &-0.20&0.936&0.932&0.918&0.919 \\
			{60}  &0    &0.086&0.032&0.056&0.115         &0    &0.936&0.918&0.913&0.931\\
			{}    &0.20  &0.085&0.031&0.062&0.127        &0.20  &0.933&0.896&0.913&0.937 \\ \hline 
			
			{}    &-0.20&0.059&0.024&0.036&0.067          &-0.20&0.948&0.945&0.940&0.931\\
			{120} &0    &0.059&0.024&0.038&0.073          &0    &0.947&0.936&0.936&0.936\\
			{}    &0.20 &0.058&0.022&0.042&0.083          &0.20 &0.944&0.926&0.933&0.942\\ \hline 
			
			{}    &-0.20&0.047&0.019&0.028&0.052          &-0.20&0.949&0.940&0.946&0.941\\
			{180} &0    &0.047&0.017&0.031&0.058          &0    &0.945&0.939&0.943&0.945\\
			{}    &0.20 &0.046&0.018&0.033&0.066          &0.20 &0.948&0.930&0.940&0.946\\ \hline 
		\end{tabular}
	\end{center}
\end{table} 

\begin{table}[!htbp]
	% \scalefont{0.75}
	\begin{center}
		\caption{Empirical mean and absolute bias of the MLEs under Scenario 3.}
                \vspace{0.2cm}
		\label{tabela1cenario3}
		\begin{tabular}{@{}c|c|cccc|c|cccc@{}}
			\hline
			\multicolumn{6}{@{}c|}{MLE}& 
			\multicolumn{5}{@{}c}{BIAS} \\ \hline
			$n$ & $\gamma$ & $\hat{\beta}_0$ & $\hat{\beta}_1$ & $\hat{\alpha}$ & $\hat{\gamma}$ & $\gamma$ & $\hat{\beta}_0$ & $\hat{\beta}_1$ & $\hat{\alpha}$ & $\hat{\gamma}$ \\ \hline
			{}    &-0.20&1.041&1.968&1.125&-0.218         &-0.20&0.041&-0.032&-0.075&-0.018 \\
			{60}  &0    &1.018&2.000&1.120&-0.009       &0    &0.018&0.000&-0.080&-0.009 \\
			{}    &0.20 &1.001&2.023&1.118&0.199         &0.20 &0.001&0.023&-0.082&-0.001 \\ \hline 
			{}    &-0.20&1.010&1.997&1.141&-0.203        &-0.20&0.010&-0.003&-0.059&-0.003 \\
			{120}  &0   &1.003&2.009&1.140&-0.003      &0   &0.003&0.009&-0.060&-0.003 \\
			{}    &0.20 &0.998&2.017&1.141&0.196        &0.20&-0.002&0.017&-0.059&-0.004 \\ \hline 
			{}    &-0.20&1.001&2.004&1.147&-0.200        &-0.20&0.001&0.004&-0.053&0.000 \\
			{180}  &0   &0.999&2.009&1.147&-0.003       &0   &-0.001&0.009&-0.053&-0.003 \\
			{}    &0.20 &0.999&2.012&1.149&0.194        &0.20 &-0.001&0.012&-0.051&-0.006\\ \hline 
		\end{tabular}
	\end{center}
\end{table}

\begin{table}[!htbp]
	% \scalefont{0.75}
	\begin{center}
		\caption{Empirical root mean square error (RMSE) and coverage probability (CP) of the MLEs under Scenario 3.}
        \vspace{0.2cm}
		\label{tabela2cenario3}
		\begin{tabular}{@{}c|c|cccc|c|cccc@{}}
			\hline
			\multicolumn{6}{@{}c|}{RMSE}& 
			\multicolumn{5}{@{}c}{CP} \\ \hline
			$n$ & $\gamma$ & $\hat{\beta}_0$ & $\hat{\beta}_1$ & $\hat{\alpha}$ & $\hat{\gamma}$ & $\gamma$ & $\hat{\beta}_0$ & $\hat{\beta}_1$ & $\hat{\alpha}$ & $\hat{\gamma}$ \\ \hline
			{}    &-0.20&0.326&0.500&0.144&0.126         &-0.20&0.916&0.918&0.856&0.923 \\
			{60}  &0    &0.328&0.505&0.150&0.132         &0    &0.925&0.927&0.857&0.925\\
			{}    &0.20  &0.312&0.469&0.159&0.142        &0.20  &0.922&0.921&0.863&0.930 \\ \hline 
			
			{}    &-0.20&0.214&0.335&0.103&0.074          &-0.20&0.940&0.939&0.856&0.946\\
			{120} &0    &0.218&0.342&0.107&0.080         &0    &0.940&0.938&0.860&0.942\\
			{}    &0.20 &0.207&0.318&0.113&0.089         &0.20 &0.939&0.934&0.877&0.942\\ \hline 
			
			{}    &-0.20&0.166&0.268&0.087&0.057          &-0.20&0.941&0.938&0.845&0.946\\
			{180} &0    &0.171&0.277&0.089&0.064          &0    &0.939&0.935&0.856&0.944\\
			{}    &0.20 &0.164&0.256&0.094&0.071         &0.20 &0.940&0.929&0.873&0.946\\ \hline 
		\end{tabular}
	\end{center}
\end{table} 

\section{Diagnostic analysis}\label{influencialocal}

\subsection{Local influence}\label{secavainflocal}
The methodology of local influence, proposed by \cite{cook1986}, is a diagnostic technique that enables to carefully detect the existence of points that, when submitted to small disturbances, cause disproportionate variation in the inferential results obtained according to a postulated model.  It is based on the geometric concept of curvature to investigate the local behavior of the offset-by-likelihood graph.
Let $\ell({\mbox{\boldmath{$\theta$}}})$ be the log-likelihood function associated with the postulated statistical model, where ${\mbox{\boldmath{$\theta$}}}=(\theta_1,...,\theta_p)^{\top}\in\Theta \subset {\R}^p$ is a $p\times 1$ vector of unknown parameters. Once the perturbation scheme is fixed, let $\ell({\mbox{\boldmath{$\theta$}}}|{\mbox{\boldmath{$\omega$}}})$ be the log-likelihood function associated with the statistical model perturbed by a vector  ${\mbox{\boldmath{$\omega$}}}\in U$. Here, it is assumed that $\ell({\mbox{\boldmath{$\theta$}}}|{\mbox{\boldmath{$\omega$}}})$ is differentiable in $({\mbox{\boldmath{$\theta$}}}^{\top},{\mbox{\boldmath{$\omega$}}}^{\top})$, that is, it has partial derivatives of all orders and they are continuous. Furthermore, it is assumed that a vector  ${\mbox{\boldmath{$\omega$}}}_0$ exists, such that   $\ell({\mbox{\boldmath{$\theta$}}})=\ell({\mbox{\boldmath{$\theta$}}}|{\mbox{\boldmath{$\omega$}}}_0)$ for all  ${\mbox{\boldmath{$\theta$}}}\in\Theta$, called non-disturbance vector, and we assume $q=n$. Let ${\mbox{\boldmath{$\hat{\theta}$}}}$ and ${{\mbox{\boldmath{$\hat{\theta}$}}}}_{{\mbox{\boldmath{$\omega$}}}}$ be the maximum likelihood estimators for ${\mbox{\boldmath{$\theta$}}}$ according to the postulated and perturbed models, respectively. In order to evaluate the influence when ${\mbox{\boldmath{$\omega$}}}$ varies in $\U$, \cite{cook1986} considers the likelihood displacement function defined by
\begin{equation}
g:U \longmapsto {\R}, \quad g({\mbox{\boldmath{$\omega$}}})= 2[\ell({\mbox{\boldmath{$\hat{\theta}$}}})-\ell({{\mbox{\boldmath{$\hat{\theta}$}}}}_{\mbox{\boldmath{$\omega$}}})].
\label{f:afastamento}
\end{equation}Then, the author mentions that the graph of $g({\mbox{\boldmath{$\omega$}}})$ versus ${\mbox{\boldmath{$\omega$}}}$ has relevant information about the perturbation scheme used, and highlights the importance of identifying it as a surface  $G\subset\R^{n+1}$ formed by the points
\begin{equation}
{\mbox{\boldmath{$\phi$}}}({\mbox{\boldmath{$\omega$}}})=
\left(
\begin{array}{c}
{\mbox{\boldmath{$\omega$}}}\\
g({\mbox{\boldmath{$\omega$}}})
\end{array}
\right), \quad {\mbox{\boldmath{$\omega$}}}\in U,
\label{s:graficodeinfluencia}
\end{equation} calling it an {\it influence graph}.

An operational formula, deduced by \cite{cook1986}, for calculating the normal curvature of $G$  at a critical critical point $\ommega_0$ of the likelihood displacement function is the following: 
\begin{equation}
\left. C_{\ele}=-2\ele^{\top}(\Deltta^{\top}\ddL^{-1}\Deltta)\ele \right|_{\ttheta=\hatttheta,\ommega=\ommega_0},
\label{curvaturanormaloperacional}
\end{equation}
where $\Deltta=(\Deltta_{ij})_{p\times n}$ is the matrix whose terms are given by
\begin{equation}
\Deltta_{ij}=\frac{\partial^2 \ell(\ttheta | \ \ommega)}{\partial \theta_i \partial \omega_j},\  i=1,...,p \mbox{ e } \ j=1,2,...,n
\label{termosmatrizdelta}
\end{equation}
evaluated at  $\ttheta=\hatttheta$ and $\ommega=\ommega_0$.  $\ddL=(\ddL_{ij})_{p\times p}$ is the matrix whose entries are given by
\begin{equation}
\ddL_{ij}=\frac{\partial^2 \ell (\ttheta | \ \ommega)}{\partial \theta_i \partial \theta_j},\  i=1,...,p \mbox{ and } \ j=1,...,p 
\label{termosmatrizddL}
\end{equation}
evaluated at  $\ttheta=\hatttheta$ and $\ommega=\ommega_0$.

Based on \cite{cook1986}, there are several ways in which the normal curvature, given by equation (\ref{curvaturanormaloperacional}), can be used to investigate the local behavior of the graphical surface of influence. In particular, he mentions that we can proceed with the analysis, paying attention to the direction in which the  normal curvature is maximum, which corresponds to the eigenvector of $-\ddot{\F}=\Deltta^{\top}(-\ddL^{-1})\Deltta$, associated with the highest eigenvalue in absolute value. \cite{cook1986} argues that the eigenvector $\ele_{\max}$ associated with $$\lambda_{\max}=\max_{\ele}{C_{\ele}}$$ indicates how to perturb the postulated model in order to obtain large local changes in the likelihood displacement function and suggests that we investigate the potential cause of this instability.

According to \cite{poon1999conformal}, the normal curvature presents some slight inconveniences when we aim at conducting a diagnostic analysis. Because, since the $\lambda_{\max}$ can assume an arbitrarily large value, the curvature will assume any real value inferior than or equal to the  $\lambda_{\max}$, causing a loss of objectivity in judging the magnitude of this curvature. Furthermore, they argue that the normal curvature is not invariant by the scale change imposed on the perturbation scheme. In order to solve these inconveniences and, consequently, to improve the technique of local influence, \cite{poon1999conformal} proposed the use of the normal curvature according to  $B_{\ele}$, which assumes values in the interval $[0,1]$, enabling the construction of reference values to judge the magnitude of $B_{\ele}$ more objectively, and it is invariant under conformal reparametrization.

At a critical point $\omega_0$ of $g$, the direction of the unit vector $\ele$ is determined using the matrices $\Delta$ and ${\ddL}$, defined in equations \eqref{termosmatrizdelta} and \eqref{termosmatrizddL} respectively, through the following operational formula for $B_{\ele}$:
% In terms of the matrices $\Deltta$ and  ${\ddL}$, given by equations \eqref{termosmatrizdelta}) and \eqref{termosmatrizddL}, and, respectively, an operational formula for the computation of $B_{\ele}$, at a critical point $\ommega_0$ of $g$, the direction of the unit vector $\ele$ is given by:
\begin{equation}
\left.B_{\ele}=-\frac{\ele^{\top}(\Deltta^{\top}\ddL^{-1}\Deltta)\ele}{\sqrt{\mathrm{tr}[(\Deltta^{\top}\ddL^{-1}\Deltta)^2]}}\right|_{\ttheta=\hatttheta,\ommega=\ommega_0}.
\label{cnconformeoperacional}
\end{equation}

Let us consider a direction  $\mathbf{v}_i\in \mathcal{F}=\{\mathbf{v}_1,\mathbf{v}_2,...,\mathbf{v}_n\}$, which is an orthonormal basis formed by the eigenvectors of $\ddot{\F}$ whose associated eigenvalues are $\lambda_1,...,\lambda_n.$, it follows that
\begin{equation}
B_{\mathbf{v}_i}=\frac{1}{\sqrt{\sum_{k=1}^{n}\lambda_{k}^{2}}}\frac{\mathbf{v}_{i}^{\top}\lambda_i\mathbf{v}_i}{\mathbf{v}_{i}^{\top}\mathbf{v}_i}= \frac{\lambda_i}{\sqrt{\sum_{k=1}^{n}\lambda_{k}^{2}}}=\lambda_{i}^{\ast}.
\label{cnconautovetor}
\end{equation}
This implies that if  $B_{\mathbf{v}_i}=B$ for every eigenvector $\mathbf{v}_i\in \mathcal{F}$, then $B=\frac{1}{\sqrt{n}}$. If we use  $\frac{1}{\sqrt{n}}$in the construction of reference values for the curvatures associated with the eigenvectors, we make our decision to judge the magnitude of  $B_{\mathbf{v}_i}$ more objectively. In this context, \cite{poon1999conformal} establish the following criteria: given an integer $q$, with $1\leq q <\sqrt{n}$,  an eigenvector $\mathbf{v}$ is said to be  $q-$influential  if $|B_{\mathbf{v}}|\geq \frac{q}{\sqrt{n}}$. Equivalently, the direction $\mathbf{v}_i$ is said to be $q-$influential if $\lambda_{i}^{\ast}\geq \frac{q}{\sqrt{n}}$. Given the existence of a positive integer $k$ such that 
$\lambda_{1}^{\ast}\geq \cdots \geq \lambda_{k}^{\ast}\geq \frac{q}{\sqrt{n}}>\lambda_{k+1}^{\ast}\geq \cdots \geq \lambda_{n}^{\ast}
$, twe have exactly $k$ eigenvectors that are $q$-influential. 

With the purpose of determining which are the basic perturbation vectors that contribute for the directions associated with the $q$-influential eigenvectors to be investigated in the analysis, \cite{poonpoon2010} defined the aggregate contribution of the $j$-th basic perturbation vector by
\begin{equation}
B_{j}(q)=\sum_{i=1}^{k}{\lambda_{i}^{\ast}}a_{ji}^{2},
\label{contribuicaoagregadaq}
\end{equation}
where the positive integer $k$ is the quantity of $q$-influential eigenvectors and $(a_{1i},a_{2i},...,a_{ni})$ are the coordinates of $\mathbf{v}_i$ in the canonical basis of $\R^n   $. Hence, if $B_{j}(q)=b(q)$ for every basic perturbation vector  $\mathbf{e}_j$, then
$b(q)=\frac{1}{n}\sum_{i=1}^{k}\lambda_{i}^{\ast}$. \cite{poonpoon2010} propose the use of $b(q)$ to construct reference values
to judge the contribution of each of the basic perturbation vectors in the determination of the  $q$-influential eigenvectors.

\subsection{Perturbation scheme}
In this subsection, we derive the curvatures for some common perturbation schemes in either the model or the data.

\paragraph{Weights perturbation} Let us consider the linear regression model, defined in (\ref{mrlevbs}), the variables $Y_1,...,Y_n$ are independent and, for each $i=1,...,n$, $Y_i\sim \mbox{log-EVBS}(\alpha,\mathbf{x}_{i}^{\top}{\mbox{\boldmath{$\beta$}}}, \gamma)$. We will admit some assumptions that enable the estimation process and add properties to the MLE of the vector $\ttheta=(\bbeta^{\top}, \alpha, \gamma)^{\top}$. The first hypothesis is $\gamma>-1$, this is so that the MLE consistency hypothesis makes sense. The second is a condition imposed on the parameter $\alpha$ so that we do not have multiplicity of solutions for the system of maximum likelihood equations, more precisely, we will assume that $\alpha<2$. Finally, with the purpose that the Fisher information matrix assumes finite values, we will assume the auxiliary hypothesis $\gamma<1/4$. Our actions will address two distinct cases. Initially, we will deal with the situation where $\gamma=0$. Next, we will deal with the general case where $\gamma\neq 0$.  

\begin{enumerate}
	\item  $\gamma = 0$. In this case, the vector of unknown parameters is $\ttheta=(\bbeta^{\top}, \alpha)^{\top}$ and, for estimation purposes, given a sample  $\mathbf{y}=(y_1,...,y_n)^{\top}$ of observations of the vector $\mathbf{Y}=(Y_1,...,Y_n)^{\top}$, the perturbed log-likelihood function of $\ttheta$ corresponds to the case weighting scheme, and it can be written in the form
	\begin{equation}
	l(\ttheta | \ \ommega)=(-\log 2)\sum_{i=1}^{n}\omega_i +\sum_{i=1}^{n}\omega_i\log(\xi_{i1})
	-\sum_{i=1}^{n}\omega_i\xi_{i2}-\sum_{i=1}^{n}\omega_i\exp(-\xi_{i2}),
	\label{flvppc}
	\end{equation}
	where $\xi_{i1}=\frac{2}{\alpha }\cosh \left(\frac{y_i-\mathbf{x}_{i}^{\top}\bbeta}{2}\right)$, $\xi_{i2}=\frac{2}{\alpha}$$\senh\left(\frac{y_i-\mathbf{x}_{i}^{\top}\bbeta}{2}\right)$ and $\ommega=(\omega_1,...,\omega_n)^{\top}\in \Omega\subset \R^{n}$.
	
	Therefore, the matrix $\Deltta=(\Deltta_{ij})_{(p+1)\times n}$; $\Deltta_{ij}=\frac{\partial^2 l(\ttheta | \ \ommega)}{\partial \theta_i \partial \omega_j}$ $i=0,1,...,p$ e $j=1,2,...,n,$  evaluated at $\ttheta=\hatttheta$ and $\ommega_0=(1,1,...,1)^{\top}\in \R^n $, is given by $\Deltta=(\Deltta_{\bbeta},\Deltta_{\alpha})^{\top}$ whose entries are 
	\begin{equation}
	\Deltta_{\bbeta}=\X^{\top} \mbox{diag}\{\hat{b}_{01},\hat{b}_{02},...,\hat{b}_{0n}\}
	\end{equation}
	where $\hat{b}_{0j}=\frac{1}{2}\left(\hat{\xi}_{j1}(1-\exp(-\hat{\xi}_{j2}))-\frac{\hat{\xi}_{j2}}{\hat{\xi}_{j1}}\right),\ j=1,2,...,n.$ Here, 
    \begin{equation}
	\Deltta_{\alpha}=\mathbf{c}_{0}^{\top}=(\hat{c}_{01},\hat{c}_{02},...,\hat{c}_{0n}), 
	\end{equation}
	with $\hat{c}_{0j}=-\frac{1}{\hat{\alpha}}\left(1-\hat{\xi}_{j2}(1-\exp(-\hat{\xi}_{j2}))\right)$ , for $j=1,2,...,n$. 
	
	\item $\gamma \neq 0$ . Given a sample  $\mathbf{y}=(y_1,...,y_n)^{\top}$ of observations of the vector $\mathbf{Y}=(Y_1,...,Y_n)^{\top}$, the perturbed log-likelihood function of $\ttheta=(\bbeta^{\top},\alpha,\gamma)^{\top}$ corresponds to the scheme of case weighting, and it can be written in the form
	\begin{equation}
	l(\ttheta | \ \ommega)=-\log 2\sum_{i=1}^{n}\omega_i +\sum_{i=1}^{n}\omega_i\log(\xi_{i1})
	-(1+{1}/{\gamma})\sum_{i=1}^{n}\omega_i\log(1+\gamma \xi_{i2})-\sum_{i=1}^{n}\omega_i(1+\gamma\xi_{i2})^{-1/\gamma},
	\label{flvppc2}
	\end{equation}
	
	Consequently, the matrix $\Deltta=(\Deltta_{ij})_{(p+2)\times n}$; $\Deltta_{ij}=\frac{\partial^2 l(\ttheta | \ \ommega)}{\partial \theta_i \partial \omega_j}$ $i=0,1,...,p+1$ e $j=1,2,...,n$ , evaluated at $\ttheta=\hatttheta$ and $\ommega_0=(1,1,...,1)^{\top}\in \R^n $, is given by $\Deltta=(\Deltta_{\bbeta},\Deltta_{\alpha},\Deltta_{\gamma})^{\top}$ whose entries are given by:
	\begin{equation}
	\Deltta_{\bbeta}=\X^{\top} \mbox{diag}\{\hat{b}_1,\hat{b}_2,...,\hat{b}_n\}
	\end{equation}
	where $\hat{b}_{j}=\frac{1}{2}\left[\frac{\hat{\xi}_{j1}}{1+\hat{\gamma}\hat{\xi}_{j2}}\left(\hat{\gamma}+1-(1+\hat{\gamma}\hat{\xi}_{j2})^{-\frac{1}{\hat{\gamma}}}\right)-\frac{\hat{\xi}_{j2}}{\hat{\xi}_{j1}}\right],\ j=1,2,...,n$.
	
	\begin{equation}
	\Deltta_{\alpha}=\mathbf{c}^{\top}=(\hat{c}_{1},\hat{c}_{2},...,\hat{c}_{n}), 
	\end{equation}
	with $\hat{c}_{j}=-\frac{1}{\hat{\alpha}}+\frac{1}{\hat{\alpha}}\frac{\hat{\xi}_{j2}}{(1+\hat{\gamma}\hat{\xi}_{j2})}\left( \hat{\gamma}+1 - (1+\hat{\gamma}\hat{\xi}_{j2})^{-\frac{1}{\hat{\gamma}}}\right)$, for $j=1,2,...,n$. And,
	
	\begin{equation}
	\Deltta_{\gamma}=\mathbf{d}^{\top}=(\hat{d}_{1},\hat{d}_{2},...,\hat{d}_{n}), 
	\end{equation}
	where $\hat{d}_{j}=\frac{1}{\hat{\gamma}^2}\log(1+\hat{\gamma}\hat{\xi}_{j2})\left(1-(1+\hat{\gamma}\hat{\xi}_{j2})^{-\frac{1}{\hat{\gamma}}}\right)-\frac{1}{\hat{\gamma}} \frac{\hat{\xi}_{j2}}{(1+\hat{\gamma}\hat{\xi}_{j2})}\left( \hat{\gamma}+1 - (1+\hat{\gamma}\hat{\xi}_{j2})^{-\frac{1}{\hat{\gamma}}}\right)$, for $j=1,2,...,n$.
	
\end{enumerate}

\paragraph{Response perturbation} In this case, it is imposed on the answer $y_i$  a perturbation of the form $y_{i\omega}=y_i+\omega_is_{y}$, in which $\omega_{i} \in \R$ for all $i=1,2,...,n$ and $s_y$ is a scale factor. The exposition that we will do in this subsection follows the same steps as the previous one with similar assumptions. 
\begin{enumerate}
	\item $\gamma=0$. When considering the linear regression model defined in (\ref{mrlevbs}) and the type of perturbation described previously, the perturbed log-likelihood function of $\ttheta=(\bbeta^{\top},\alpha)^{\top}$, given a sample  $\mathbf{y}=(y_1,...,y_n)^{\top}$ of observations, is given by
	\begin{equation}
	l(\ttheta | \ \ommega)=-n\log 2 +\sum_{i=1}^{n}\log(\xi_{i1\omega_r})
	-\sum_{i=1}^{n}\xi_{i2\omega_r}-\sum_{i=1}^{n}\exp(-\xi_{i2\omega_r}),
	\label{flvpresp}
	\end{equation}
	where $\xi_{i1\omega_r}=\frac{2}{\alpha }\cosh \left(\frac{y_i+\omega_{i}s_y-\mathbf{x}_{i}^{\top}\bbeta}{2}\right)$ and
	$\xi_{i2\omega_r}=\frac{2}{\alpha}\senh\left(\frac{y_i+\omega_{i}s_y-\mathbf{x}_{i}^{\top}\bbeta}{2}\right)$ . 
	
	Thus, when we perturb the response $y_i$ as in the described scheme, the matrix  $\Deltta=(\Deltta_{ij})_{(p+1)\times n}$; $\Deltta_{ij}=\frac{\partial^2 l(\ttheta | \ \ommega)}{\partial \theta_i \partial \omega_j}$ $i=0,1,...,p$ and $j=1,2,...,n$ , evaluated at $\ttheta=\hatttheta$ and $\ommega=\ommega_0=(0,0,...,0)^{\top}\in \R^n $, is given by $\Deltta=(\Deltta_{\bbeta},\Deltta_{\alpha})^{\top}$ whose entries are 
	\begin{equation}
	\Deltta_{\bbeta}=\X^{\top} \mbox{diag}\{\hat{m}_{01},\hat{m}_{02},...,\hat{m}_{0n}\}
	\end{equation}
	where $\hat{m}_{0j}= \frac{1}{4}s_y\left[1-\left( \frac{\hat{\xi}_{j2\omega_r}}{\hat{\xi}_{j1\omega_r}}\right)^2-\hat{\xi}_{j2\omega_r}(1-\exp(-\xi_{j2\omega_r})) -\hat{\xi}_{j1\omega_r}^{2} \exp(-\hat{\xi}_{j2\omega_r}) \right]$, for $j=1,2,...,n$. Moreover,
	
	\begin{equation}
	\Deltta_{\alpha}=\mathbf{z}_{0}^{\top}=(\hat{z}_{01},\hat{z}_{02},...,\hat{z}_{0n}), 
	\end{equation}
	
	with $\hat{z}_{0j}= -\frac{1}{2\hat{\alpha}}s_y\hat{\xi}_{j1\omega_r}[1-\exp(-\hat{\xi}_{j2\omega_r})(1-\hat{\xi}_{j2\omega_r})]  $, for $j=1,2,...,n$. 
	
	\item $\gamma \neq 0$.  Given a sample  $\mathbf{y}=(y_1,...,y_n)^{\top}$ of realizations of the vector $\mathbf{Y}=(Y_1,...,Y_n)^{\top}$, the perturbed log-likelihood function of $\ttheta=(\bbeta^{\top},\alpha,\gamma)^{\top}$ corresponding to the perturbation scheme in the response is
	\begin{equation}
	l(\ttheta | \ \ommega)=-n\log 2+\sum_{i=1}^{n}\log(\xi_{i1\omega_r})
	-(1+{1}/{\gamma})\sum_{i=1}^{n}\log(1+\gamma \xi_{i2\omega_r})-\sum_{i=1}^{n}(1+\gamma\xi_{i2\omega_r})^{-1/\gamma},
	\label{flvppres}
	\end{equation}
	where $\xi_{i1\omega_r}=\frac{2}{\alpha }\cosh \left(\frac{y_i+\omega_is_y-\mathbf{x}_{i}^{\top}\bbeta}{2}\right)$ and $\xi_{i2\omega_r}=\frac{2}{\alpha}$$\senh\left(\frac{y_i+\omega_is_y-\mathbf{x}_{i}^{\top}\bbeta}{2}\right)$ .
	
	Therefore, the matrix $\Deltta=(\Deltta_{ij})_{(p+2)\times n}$; $\Deltta_{ij}=\frac{\partial^2 l(\ttheta | \ \ommega)}{\partial \theta_i \partial \omega_j}$ $i=0,1,...,p+1$ e $j=1,2,...,n$, evaluated at  $\ttheta=\hatttheta$ and $\ommega=\ommega_0=(0,0,...,0)^{\top}\in \R^n $, is given by $\Deltta=(\Deltta_{\bbeta},\Deltta_{\alpha},\Deltta_{\gamma})^{\top}$ whose entries are:
	\begin{equation}
	\Deltta_{\bbeta}=\X^{\top}\mbox{diag}\{\hat{m}_1,\hat{m}_2,...,\hat{m}_n\},
	\end{equation}
	where
	\begin{eqnarray*}
		\hat{m}_{j}&=&
		\frac{1}{4}s_y[1-\frac{\hat{\xi}_{j1\omega_r}^{2}}{\hat{\xi}_{j2\omega_r}^{2}}-\frac{\hat{\xi}_{j2\omega_r}}{1+\hat{\gamma}\hat{\xi}_{j2\omega_r}}(1+\hat{\gamma}-(1+\hat{\gamma}\hat{\xi}_{j2\omega_r})^{-\frac{1}{\hat{\gamma}}})+\\
		&+&(1+\hat{\gamma})\left(\frac{\hat{\xi}_{j1\omega_r}}{1+\hat{\gamma}\hat{\xi}_{j2\omega_r}}\right)^2(\hat{\gamma} -(1+\hat{\gamma}\hat{\xi}_{j2\omega_r})^{-\frac{1}{\hat{\gamma}}})], \ \text{for} \ j=1,...,n.
	\end{eqnarray*} 
	
	\begin{equation}
	\Deltta_{\alpha}=\mathbf{z}^\top=(\hat{z}_1,...,\hat{z}_n), 
	\end{equation}
	with 
	\begin{eqnarray*}
		\hat{z}_j&=& -\frac{1}{2\hat{\alpha}}s_y\left( \frac{\hat{\xi}_{j1\omega_r}}{1+\hat{\gamma}\hat{\xi}_{j2\omega_r}}\right)\left(1+\hat{\gamma}-(1+\hat{\gamma}\hat{\xi}_{j2\omega_r})^{-\frac{1}{\hat{\gamma}}}\right)+\\
		{}&+&\frac{1}{2\hat{\alpha}}s_y(1+\hat{\gamma})\left(\frac{\hat{\xi}_{j1\omega_r}}{1+\hat{\gamma}\hat{\xi}_{j2\omega_r}}\right)\left(\frac{\hat{\xi}_{j2\omega_r}}{1+\hat{\gamma}\hat{\xi}_{j2\omega_r}}\right)\left(\hat{\gamma}-(1+\hat{\gamma}\hat{\xi}_{j2\omega_r})^{-\frac{1}{\hat{\gamma}}}\right),
	\end{eqnarray*}
	$j=1,2,...,n$. Finally, we have
	
	\begin{equation}
	\Deltta_{\gamma}=\mathbf{v}^\top=(\hat{v}_1,...,\hat{v}_n),
	\end{equation}
	where
	\begin{eqnarray*}
		\hat{v}_j&=& \frac{1}{2}s_y\left( \frac{\hat{\xi}_{j1\omega_r}}{1+\hat{\gamma}\hat{\xi}_{j2\omega_r}}\right)\left(1-\frac{1}{\hat{\gamma}^2}\log(1+\hat{\gamma}\hat{\xi}_{j2\omega_r})\right)+\\
		&-&\frac{1}{2}s_y(1+\hat{\gamma})\left(\frac{\hat{\xi}_{j1\omega_r}}{1+\hat{\gamma}\hat{\xi}_{j2\omega_r}}\right)\left(\frac{\hat{\xi}_{j2\omega_r}}{1+\hat{\gamma}\hat{\xi}_{j2\omega_r}}\right)\left(1-\frac{1}{\hat{\gamma}}(1+\hat{\gamma}\hat{\xi}_{j2\omega_r})^{-\frac{1}{\hat{\gamma}}}\right),
	\end{eqnarray*}
	for  $j=1,2,...,n$.
\end{enumerate}

\paragraph{Explanatory variable perturbation}

Let us assume that in the model defined by equation (\ref{mrlevbs}) there is a continuous control parameter denoted by $\mathbf{x}_t$. We impose over $\mathbf{x}_t$ an additive perturbation of the form $x_{it\omega}=x_{it}+\omega_{i}s_x$, where  $t\in \{0,1,...,p-1\}$, and $s_x$ is a scale factor, e.g. the corresponding standard deviation to the observations regarding the variable $\mathbf{x}_t$ . Next, we will obtain the expressions that allow the calculation of curvatures according to this scheme of perturbation. 
\begin{enumerate}
	\item $\gamma=0$. By imposing a perturbation on a single continuous explanatory variable, according to the scheme described above, the perturbed log-likelihood function of $\ttheta=(\bbeta^{\top},\alpha)^{\top}$, given a sample  $\mathbf{y}=(y_1,...,y_n)^{\top}$ of observations, is   
	\begin{equation}
	l(\ttheta | \ \ommega)=-n\log 2 +\sum_{i=1}^{n}\log(\xi_{i1\omega_c})
	-\sum_{i=1}^{n}\xi_{i2\omega_c}-\sum_{i=1}^{n}\exp(-\xi_{i2\omega_c}),
	\label{flvpcov}
	\end{equation}
	where $\xi_{i1\omega_c}=\frac{2}{\alpha }\cosh \left(\frac{y_i-\mathbf{x}_{i}^{\top}\bbeta-\beta_i\omega_is_x}{2}\right)$ and
	$\xi_{i2\omega_c}=\frac{2}{\alpha}\senh\left(\frac{y_i-\mathbf{x}_{i}^{\top}\bbeta-\beta_t\omega_is_x}{2}\right)$.  
	
	Hence, the matrix $\Deltta=(\Deltta_{ij})_{(p+1)\times n}$; $\Deltta_{ij}=\frac{\partial^2 l(\ttheta | \ \ommega)}{\partial \theta_i \partial \omega_j}$ $i=0,1,...,p$ and $j=1,2,...,n,$ corresponding to the
	perturbation scheme on the continuous covariate $\mathbf{x}_t$, evaluated at $\ttheta=\hatttheta$ and $\ommega=\ommega_0=(0,0,...,0)^{\top}\in \R^n $, is 
	$\Deltta=(\Deltta_{\bbeta},\Deltta_{\alpha})^{\top}$ whose entries are
	\begin{equation}
	\Deltta_{\bbeta}=\X^{\top} \mbox{diag}\{\hat{\kappa}_{01},\hat{\kappa}_{02},...,\hat{\kappa}_{0n}\}
	\end{equation}
	So, if $i\neq t$, then 
	\begin{eqnarray*} 
		\hat{\kappa}_{0j}&=& \frac{1}{4}s_x\hat{\beta}_t\left[1-\left( \frac{\hat{\xi}_{j2\omega_c}}{\hat{\xi}_{j1\omega_c}}\right)^2-\hat{\xi}_{j2\omega_c}(1-\exp(-\hat{\xi}_{j2\omega_c})) -\hat{\xi}_{j1\omega_c}^{2} \exp(-\hat{\xi}_{j2\omega_c}) \right],
	\end{eqnarray*}for $j=1,2,...,n$. On the other hand, if  $i=t$, then
	\begin{eqnarray*} 
		\hat{\kappa}_{0j}&=& \frac{1}{4}s_x\hat{\beta_t}\left[1-\left( \frac{\hat{\xi}_{j2\omega_c}}{\hat{\xi}_{j1\omega_c}}\right)^2-\hat{\xi}_{j2\omega_c}(1-\exp(-\hat{\xi}_{j2\omega_c})) -\hat{\xi}_{j1\omega_c}^{2} \exp(-\hat{\xi}_{j2\omega_c}) \right]\! \!+\\
		&+&\left(-\frac{1}{2}\right)  s_x\left[\frac{\hat{\xi}_{j2\omega_c}}{\hat{\xi}_{j2\omega_c}}-\hat{\xi}_{j1\omega_c}(1-\exp(-\hat{\xi}_{j2\omega_c}))\right],
	\end{eqnarray*}
	for $j=1,2,...,n,$ and,
	\begin{equation}
	\Deltta_{\alpha}=\mathbf{\tau}_{0}^{\top}=(\hat{\tau}_{01},\hat{\tau}_{02},...,\hat{\tau}_{0n}), 
	\end{equation}
	
	with $\hat{\tau}_{0j}= -\frac{1}{2\hat{\alpha}}s_x\hat{\xi}_{j1\omega_c}[1-\exp(-\hat{\xi}_{j2\omega_c})(1-\hat{\xi}_{j2\omega_c})]$, for $j=1,2,...,n$.
	
	\item $\gamma \neq 0$.  The perturbed log-likelihood function of $\ttheta=(\bbeta^{\top},\alpha,\gamma)^{\top}$ corresponding to the perturbation scheme on an explanatory variable is
	\begin{equation}
	l(\ttheta | \ \ommega)=-n\log 2+\sum_{i=1}^{n}\log(\xi_{i1\omega_c})
	-(1+{1}/{\gamma})\sum_{i=1}^{n}\log(1+\gamma \xi_{i2\omega_c})-\sum_{i=1}^{n}(1+\gamma\xi_{i2\omega_c})^{-1/\gamma},
	\label{flvppcov}
	\end{equation}
	where $\xi_{i1\omega_c}=\frac{2}{\alpha }\cosh \left(\frac{y_i-\mathbf{x}_{i}^{\top}\bbeta -\beta_t\omega_is_x}{2}\right)$, $\xi_{i2\omega_c}=\frac{2}{\alpha}$$\senh\left(\frac{y_i-\mathbf{x}_{i}^{\top}\bbeta-\beta_t\omega_is_x}{2}\right)$.
	
	Hence, the matrix $\Deltta=(\Deltta_{ij})_{(p+2)\times n}$; $\Deltta_{ij}=\frac{\partial^2 l(\ttheta | \ \ommega)}{\partial \theta_i \partial \omega_j}$ $i=0,1,...,p+1$ e $j=1,2,...,n,$  evaluated at $\ttheta=\hatttheta$ and$\ommega=\ommega_0=(0,0,...,0)^{\top}\in \R^n $, is  $$\Deltta=(\Deltta_{\bbeta},\Deltta_{\alpha},\Deltta_{\gamma})^{\top},$$ whose entries are
	\begin{equation}
	\Deltta_{\bbeta}=\X^{\top}\mbox{diag}\{\hat{\kappa}_1,\hat{\kappa}_2,...,\hat{\kappa}_n\}.
	\end{equation}
	
	So, if $i\neq t$, then
	\begin{eqnarray*}
		\hat{\kappa}_j&=&
		\frac{1}{4}s_x\hat{\beta}_t[1-\frac{\hat{\xi}_{j1\omega_c}^{2}}{\hat{\xi}_{j2\omega_c}^{2}}-\frac{\hat{\xi}_{j2\omega_c}}{1+\hat{\gamma}\hat{\xi}_{j2\omega_c}}(1+\hat{\gamma}-(1+\hat{\gamma}\hat{\xi}_{j2\omega_c})^{-\frac{1}{\hat{\gamma}}})+ \\
		&+& (1+\hat{\gamma})\left(\frac{\hat{\xi}_{j1\omega_c}}{1+\hat{\gamma}\hat{\xi}_{j2\omega_c}}\right)^2(\hat{\gamma} -(1+\hat{\gamma}\hat{\xi}_{j2\omega_c})^{-\frac{1}{\hat{\gamma}}})].
	\end{eqnarray*}
	
	However, if $i=t$, then 
	\begin{eqnarray*}
		\hat{\kappa}_j&=&
		\frac{1}{4}s_x\hat{\beta}_t[1-\frac{\hat{\xi}_{j1\omega_c}^{2}}{\hat{\xi}_{j2\omega_c}^{2}}-\frac{\hat{\xi}_{j2\omega_c}}{1+\hat{\gamma}\hat{\xi}_{j2\omega_c}}(1+\hat{\gamma}-(1+\hat{\gamma}\hat{\xi}_{j2\omega_c})^{-\frac{1}{\hat{\gamma}}})+ \\
		&+& (1+\hat{\gamma})\left(\frac{\hat{\xi}_{j1\omega_c}}{1+\hat{\gamma}\hat{\xi}_{j2\omega_c}}\right)^2(\hat{\gamma} -(1+\hat{\gamma}\hat{\xi}_{j2\omega_c})^{-\frac{1}{\hat{\gamma}}})]+\\
		&-& \frac{1}{2}s_x\left[\frac{\hat{\xi}_{j2\omega_c}}{\hat{\xi}_{j1\omega_c}}-\frac{\hat{\xi}_{j1\omega_c}}{1+\hat{\gamma}\hat{\xi}_{j2\omega_c}}\left(1+\hat{\gamma} -(1+\hat{\gamma}\hat{\xi}_{j2\omega_c})^{-1/\hat{\gamma}}\right)\right],
	\end{eqnarray*}
	with $i=0,1,...,p-1$  and $j=1,2,...,n$.
	
	Besides,
	\begin{equation}
	\Deltta_{\alpha}=\mathbf{\tau}^\top=(\hat{\tau}_1,...,\hat{\tau}_n), 
	\end{equation}
	with 
	\begin{eqnarray*}
		\hat{\tau}_j
		&=&
		-\frac{1}{2\hat{\alpha}}s_x\hat{\beta_t}\left( \frac{\hat{\xi}_{j1\omega_c}}{1+\hat{\gamma}\hat{\xi}_{j2\omega_c}}\right)\left(1+\hat{\gamma}-(1+\hat{\gamma}\hat{\xi}_{j2\omega_c})^{-\frac{1}{\hat{\gamma}}}\right)+ \\
		&+&\frac{1}{2\hat{\alpha}}s_x\hat{\beta}_t(1+\hat{\gamma})\left(\frac{\hat{\xi}_{j1\omega_c}}{1+\hat{\gamma}\hat{\xi}_{j2\omega_c}}\right)\left(\frac{\hat{\xi}_{j2\omega_c}}{1+\hat{\gamma}\hat{\xi}_{j2\omega_c}}\right)\left(\hat{\gamma}-(1+\hat{\gamma}\hat{\xi}_{j2\omega_c})^{-\frac{1}{\hat{\gamma}}}\right),
	\end{eqnarray*}
	for $j=1,...,n$.
	
	Finally, we have
	\begin{equation}
	\Deltta_{\gamma}=\mathbf{u}^\top=(\hat{u}_1,...,\hat{u}_n),
	\end{equation}
	where
	\begin{eqnarray*}
		\hat{u}_j&=&
		\frac{1}{2}s_x\hat{\beta}_t\left( \frac{\hat{\xi}_{j1\omega_r}}{1+\hat{\gamma}\hat{\xi}_{j2\omega_r}}\right)\left(1-\frac{1}{\hat{\gamma}^2}\log(1+\hat{\gamma}\hat{\xi}_{j2\omega_r})\right)+\\
		&-&\frac{1}{2}s_x\hat{\beta}_t(1+\hat{\gamma})\left(\frac{\hat{\xi}_{j1\omega_r}}{1+\hat{\gamma}\hat{\xi}_{j2\omega_r}}\right)\left(\frac{\hat{\xi}_{j2\omega_r}}{1+\hat{\gamma}\hat{\xi}_{j2\omega_r}}\right)\left(1-\frac{1}{\hat{\gamma}}(1+\hat{\gamma}\hat{\xi}_{j2\omega_r})^{-\frac{1}{\hat{\gamma}}}\right),
	\end{eqnarray*}
	for  $j=1,2,...,n$.
	
\end{enumerate}

\section{Application}\label{aplicacaodadosreais}

To illustrate the physical relevance of the methodology, we analyze a set of meteorological data corresponding to turbulent atmospheric fluctuations. We consider a set of data recorded by Station A-868 (Latitude -26.95083, Longitude -48.76194) in Itajaí, Santa Catarina, Brazil. The wind gust data were obtained from the National Institute of Meteorology (INMET) at \url{https://portal.inmet.gov.br}. We collected the maximum monthly wind speed (m/s) and the daily average atmospheric pressure (mb) corresponding to the day on which the maximum gust was recorded, from July 2010 to October 2020, resulting in a sample of size $n=124.$ 
Data are also available at \url{https://github.com/Raydonal/Extreme-value-BS-model}.  The descriptive statistics in Table \ref{stast} show that the distribution of the maximum monthly wind speed is positively skewed, characteristic of many extreme value phenomena. Furthermore, Figure \ref{gdispersaoafc}(b), which shows the sample autocorrelation function, supports the hypothesis of independence between the observed values, a common assumption in block-maxima approaches. 

\begin{table}[!htb]
	% \scalefont{1}
	\begin{center}
		% \begin{footnotesize}
		\caption{\label{stast}Descriptive statistics for the monthly maximum wind speed (m/s) in Itajaí-Brazil.}
        \vspace{0.2cm}
		\label{tabestdescrivasitaji}
		\begin{tabular}{ccccccccc}
			\hline
			n& Minimum&   Median&    Mean &    Maximum & SD& Skewness& Kurtosis \\
			\hline
			$124$& $8.20$&  $14.30$&   $14.73$&   $33.90$&$ 3.63$& $1.41$ & $5.01$ \\ 
			\hline
		\end{tabular}
		% \end{footnotesize}
	\end{center}
\end{table}

\begin{figure}[!htbp]
	\center
	\subfigure[]{\includegraphics[width=0.48\textwidth]{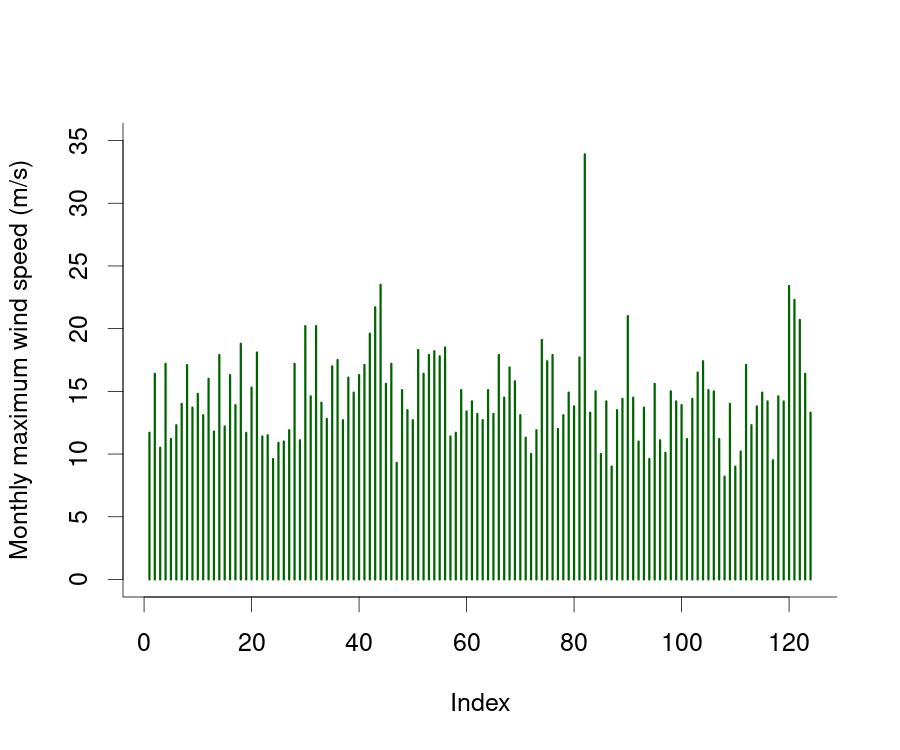} \label{fig3:a}} 
     \vspace{-0.5cm}
	\subfigure[]{\includegraphics[width=0.48\textwidth]{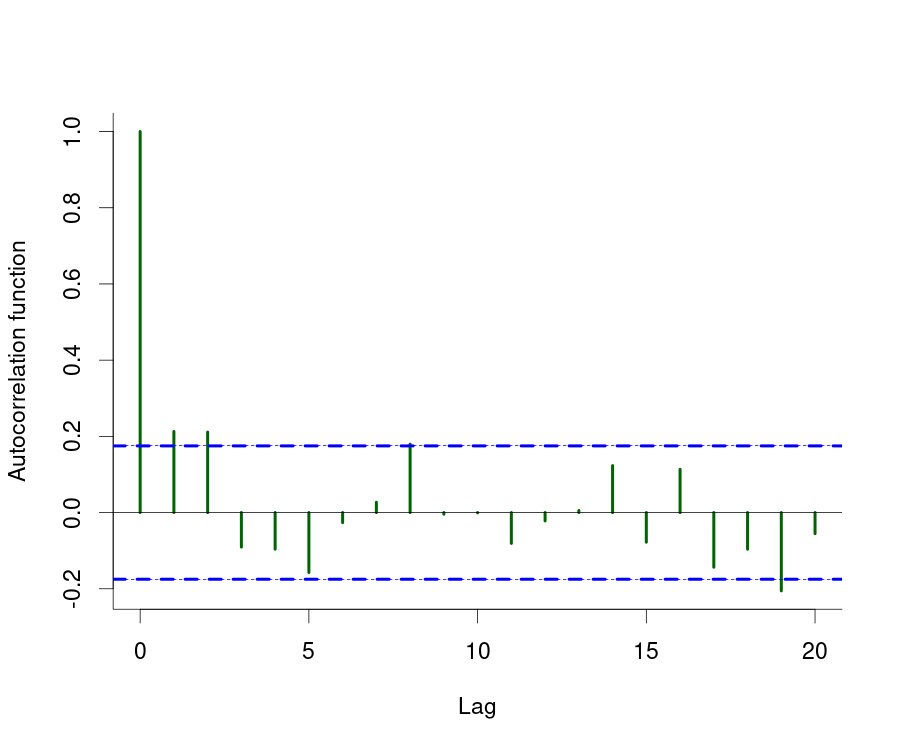} \label{fig3:b}}\\
	\caption{\label{gdispersaoafc} Scatter-plot of data (a) and the sample autocorrelation function of the monthly maximum wind speed (b)}
\end{figure}

In our analysis, we postulate an EVBS regression model where the response variable $Y$ is the natural logarithm of the maximum monthly wind speed $T$, and the explanatory variable (control parameter) is the daily average atmospheric pressure $X$. More precisely, we consider the model
\begin{equation}
Y_i=\log(T_i)=\beta_0+\beta_1x_{i}+\varepsilon_i, \quad i=1,...,124,
\label{modeloitajai}
\end{equation}
where $\varepsilon_i=\log(\delta_i)\sim \mbox{log-EVBS}(\alpha,0,\gamma)$. A preliminary fit yielded the maximum likelihood estimates (MLE) for the parameters, along with their standard errors (SE) and the p-values for the regression coefficients, shown in Table \ref{estcoefitajaipriori}. We note that the coefficient associated with the atmospheric pressure is significant at the $5\%$ level. 
\begin{table}[!htb]
	\begin{center}
		%		\begin{footnotesize}
		\caption{MLE for the model parameters with their respective standard errors and p-values for the regression coefficients.}
        \vspace{0.2cm}
		\label{estcoefitajaipriori}
		\begin{tabular}{c|c|c|c|c}
			\hline
			& $\hat{\beta}_0$ & $\hat{\beta}_1$ &$\hat{\alpha}$ & $\hat{\gamma}$ \\			
			\hline
			MLE & $25.5148$ & $-0.0227$ &$0.1857$ & $-0.1551$ \\
			\hline
			SE & $3.3338$& $0.0033$& $0.0127$ & $0.0472$ \\
			\hline
			$p$-value& $<0.0001$ &$<0.0001$ &------& ------ \\
			\hline
		\end{tabular}
		%		\end{footnotesize}
	\end{center}
\end{table}

Regarding the adequacy of the assumed distribution for the response variable, we investigated the behavior of the randomized quantile residuals, which, according to \cite{leiva2016extreme}, are effective for this purpose. Figure \ref{residuoquantilico}(a) shows the normal probability plot with envelopes for the quantile residuals, suggesting no significant deviation from the hypothesis that the response follows a log-EVBS distribution. Furthermore, Shapiro-Wilk and Kolmogorov-Smirnov tests performed on the residuals, shown in Figure \ref{residuoquantilico}(b), yielded p-values of $0.3021$ and $0.4525$, respectively. Therefore, we do not reject the normality hypothesis for the residuals, confirming the adequacy of the log-EVBS model for the response distribution.
\begin{figure}[!htbp]
	\center
	\subfigure[]{\includegraphics[width=0.48\textwidth]{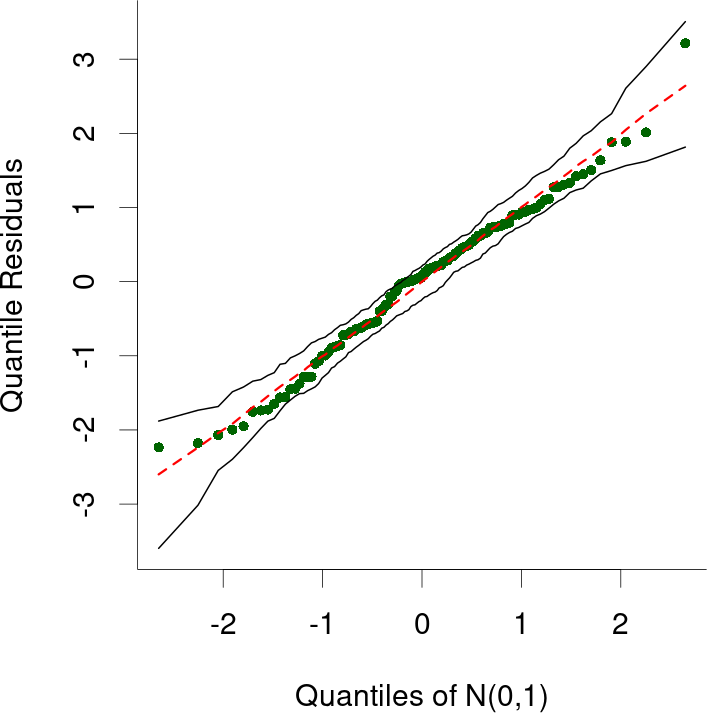} \label{fig4:a}} 
    % \hspace{-0.5cm}
	\subfigure[]{\includegraphics[width=0.48\textwidth]{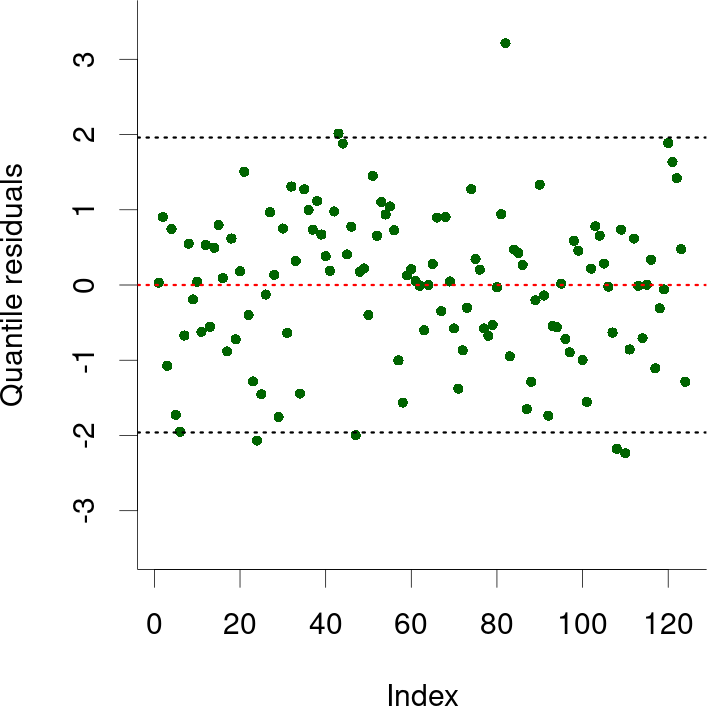} \label{fig4:b}}\\
	\caption{\label{residuoquantilico} Normal probability plot with envelope for quantile residuals (a) and index plot of quantile residuals (b)}
\end{figure}

We now proceed with a local influence analysis to identify influential observations by perturbing the model according to the case-weighting scheme. Under this scheme, the normalized eigenvalues are $0.69678$, $0.49511$, $0.44547$, and $0.26630$, as shown in Figure \ref{autovaleBjq7}(a). The largest eigenvalue corresponds to a 7-influential eigenvector, indicating the direction of maximum curvature and thus maximum model sensitivity.

Using equation (\ref{contribuicaoagregadaq}), we calculated the aggregate contribution of each basic perturbation vector to this 7-influential eigenvector. This identifies which specific observations contribute most to the direction of maximum curvature. This is illustrated in Figure \ref{autovaleBjq7}(b). Using the reference value $b(q=7)=\frac{1}{124}(0.69678) \approx 0.0056$, as suggested by \cite{poon1999conformal}, we classify observation \#82 as potentially influential and observation \#108 as marginally influential.
\begin{figure}[!htbp]
	\center
	\subfigure[]{\includegraphics[width=0.49\textwidth]{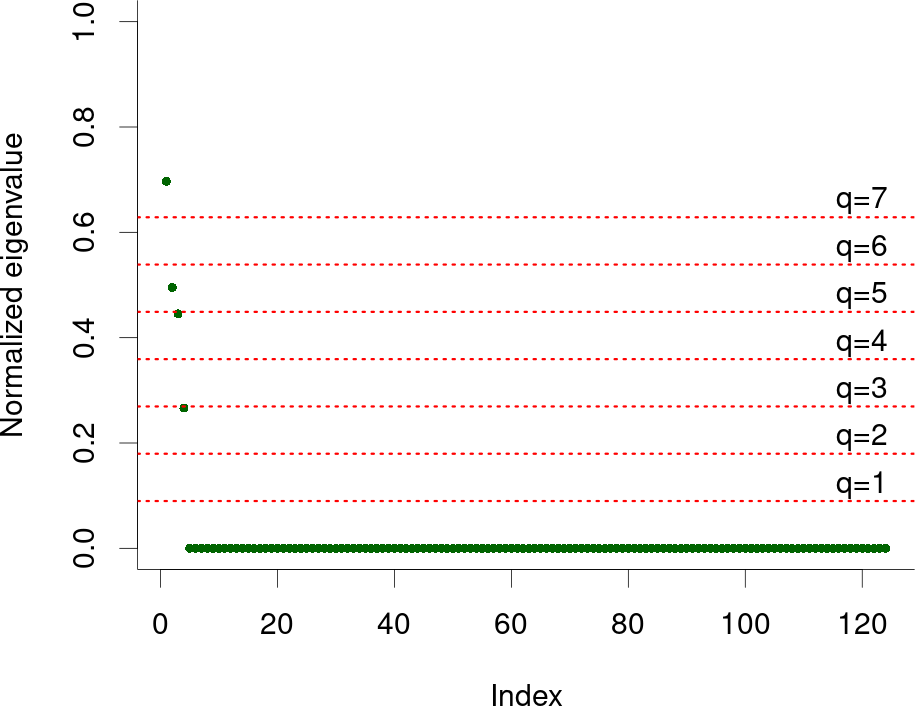} \label{fig5:a}} 
    \hspace{-0.2cm}
	\subfigure[]{\includegraphics[width=0.49\textwidth]{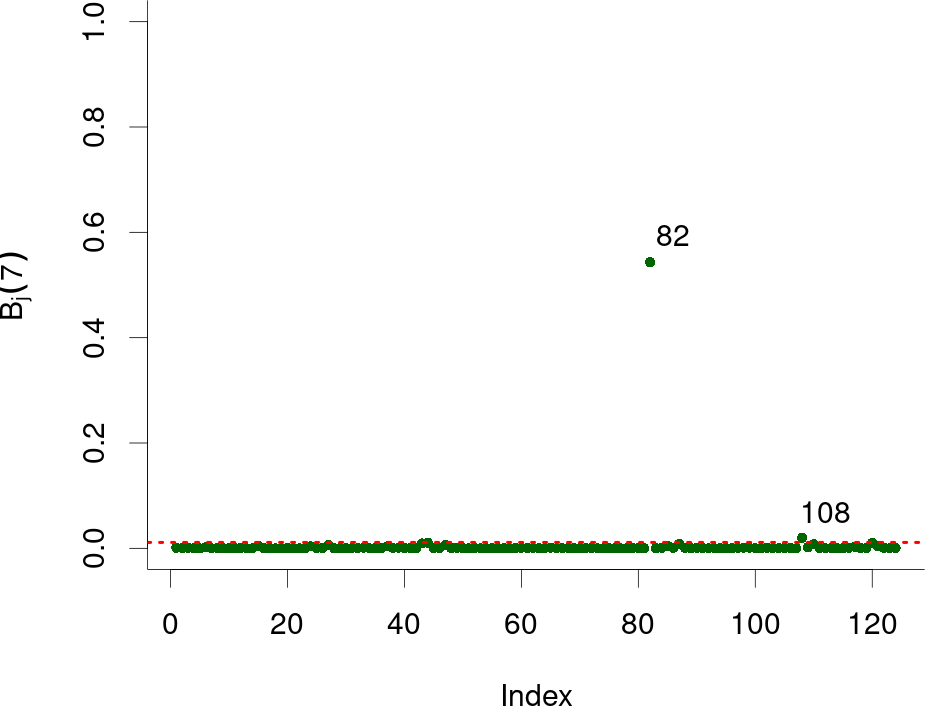} \label{fig5:b}}\\
	\caption{\label{autovaleBjq7} Normalized eigenvalues (a); aggregate contribution of the basic perturbation vectors under case-weighting (b)}
\end{figure}

To verify the impact of observations \#82 and \#108 on the parameter estimates, we refitted the model, removing each one individually. Table \ref{tabela2} shows the new maximum likelihood estimates (MLE), their respective percentage rate of change, denoted by $R_{\theta_j} = ((\hat{\theta}_{j_{(i)}} - \hat{\theta}_j) / |\hat{\theta}_j|) \times 100$, where $\hat{\theta}_{j_{(i)}}$ is the estimate of $\hat{\theta}_j$ after removing the $i$-th observation, and the p-values for the regression coefficients.

\begin{table}[!htb]
	\begin{center}
		% \begin{footnotesize}
		\caption{Maximum likelihood estimation, rate of change (in percentage), and descriptive level associated with the regression coefficients after removing an observation.}
        \vspace{0.2cm}
		\label{tabela2}
		\begin{tabular}{m{4cm}|c|ccccc}\hline
			{Dropped Observation} & Statistic&${\hat{\beta}_0}$&${\hat{\beta}_1}$&${\hat{\alpha}}$&${\hat{\gamma}}$ \\
			\hline
			{}& {MLE}  &{24.6878}& {-0.0219} &{0.1861}&{-0.2694}  \\
			{82} &{$R_{\theta_j}$}&{-3.24\%}&{+3.64\%}& {+0.25\%}& {\textcolor{red}{-73.67\%}} \\
			& {$p$-value} &{$<0.0001$}&{$<0.0001$}&---& ---\\
			\hline
			{}& {MLE}  &{24.0750}& {-0.0213} &{0.1822}&{-0.1458}  \\
			{108} &{$R_{\theta_j}$}&{-5.64\%}&{+6.29\%}&{-1.89\%}&{+6.00\%} \\
			& {$p$-value} &{$<0.0001$}&{$<0.0001$}&---& ---\\
			\hline
		\end{tabular}
		% \end{footnotesize}
	\end{center}
\end{table}

As expected, the removal of observation \#82, identified as highly influential, leads to a significant change in the parameter estimates compared to the initial fit (Table \ref{estcoefitajaipriori}). In particular, the results in Table \ref{tabela2} reveal that excluding observation \#82 causes a massive variation of $-73.67\%$ in the estimate of $\gamma$, the parameter that governs the tail behavior of the distribution. However, this exclusion has a negligible impact on the estimate of $\alpha$ and the regression coefficients $\beta_0$ and $\beta_1$, not affecting the significance of the explanatory variable. In contrast, removing observation \#108, which was only marginally influential, does not cause any significant changes in the model parameters, confirming our diagnostic. 

Investigating the physical nature of observation \#82, we found it was recorded on April 26, 2017, when the wind reached a maximum speed of $33.9$ m/s with a low daily mean atmospheric pressure of $1004.33$ mb. This data point corresponds to a documented extreme weather event. The Climate Laboratory (LabClima) at UNIVALI reported this event, noting: ``Notice the satellite image this morning showing very little cloudiness in Santa Catarina and the cold front that brought strong winds...''. Simultaneously, the State Civil Defense reported the event and highlighted its severe impact: ``Strong winds hit the municipality causing an average of 40 incidents of roof damage, falling trees and power lines... In addition, it emphasized the following human damages: 03 Deaths (ages 03, 10, and 16) caused by the falling of a power line pole resulting in electric shock.'' This demonstrates that our method successfully identified a data point corresponding to a physically significant and catastrophic event.

After the local influence analysis and the residual checks, we conclude that the postulated regression model, given by equation (\ref{modeloitajai}), is adequate for explaining the maximum monthly wind speed as a function of atmospheric pressure in Itajaí. The final model is given by
\begin{equation}
Y_i=\log(T_i)=\beta_0+\beta_1x_{i}+\varepsilon_i,  \quad i=1,...,124.
\end{equation}
where the maximum likelihood estimates of the parameters (with standard errors in parentheses) are $\hat{\beta}_0=25.5148 (3.3338)$, $\hat{\beta}_1=-0.0227 (0.0033)$, $\hat{\alpha}=0.1857 (0.0127)$, and $\hat{\gamma}=-0.1551 (0.0472)$. This results in a predictive model for the monthly maximum wind speed ($T$) as a function of the daily mean atmospheric pressure ($X$) given by
$\hat{T}=\exp(25.5148-0.0227x)$, represented by the curve in Figure \ref{modeloajustado}.

\begin{figure}[!htb]
	\begin{center}
		\includegraphics[width=0.6\textwidth]{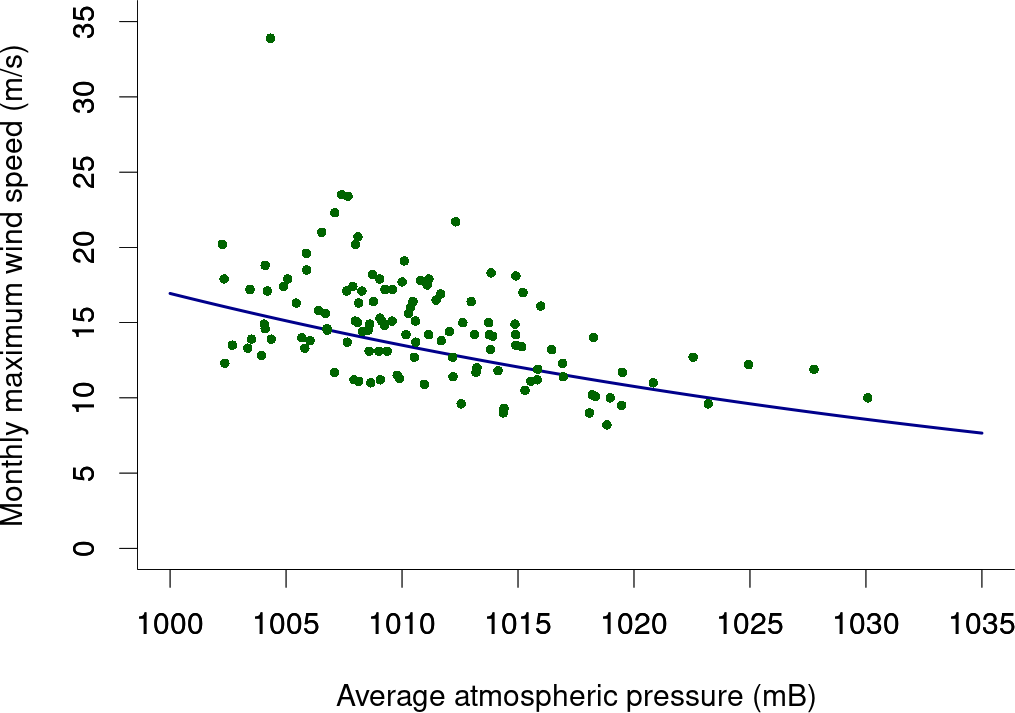}
		% \vspace{-0.2cm}	
		\caption{Scatter-plot and the fitted model of the  monthly maximum wind speed in Itajaí-Brazil}
		\label{modeloajustado}
	\end{center}	
\end{figure}

\section{Conclusion}\label{conclusao}

The task of drawing robust inferences from experimental data is a foundational challenge in statistical physics, deeply connected to the principles of entropy and irreversibility \cite{pachter2024entropy}. In this work, we have developed a formal framework for assessing the stability of extreme-value regression models, a critical step in the modeling of complex physical systems. We derived the necessary tools for a local influence analysis of the EVBS regression model under various perturbation schemes, enabling the rigorous identification of influential observations that could otherwise compromise physical interpretations.

The need for specialized, non-Gaussian models like the EVBS is not merely a statistical convenience but a physical necessity. Many complex systems, from fluid turbulence \cite{sosa2019} to systems near critical points \cite{sornette2006critical}, are known to generate observables with highly skewed, heavy-tailed distributions. The EVBS model is thus a physically motivated candidate for describing such phenomena. Our parameter estimation procedure can be framed as an inverse statistical problem \cite{nguyen2017inverse}, where the goal is to infer the underlying model parameters from noisy data. The stability analysis presented here is crucial for ensuring that the solutions to such inverse problems are robust and not mere artifacts of a few anomalous measurements.

Our analysis of real-world wind gust data demonstrated the practical power of this methodology. The local influence technique successfully pinpointed a single data point corresponding to a documented, catastrophic weather event, and we quantified its profound impact on the model's tail parameter, $\gamma$. This highlights the potential of our approach as a diagnostic tool for identifying and understanding the drivers of extreme events in natural systems.

For future work, this framework can be extended in several exciting directions. Moving beyond the static, independent-and-identically-distributed assumption of block maxima, one could incorporate dynamics by treating the error terms as a specific stochastic process \cite{jacobs2010stochastic}. This would open the door to employing more sophisticated mathematical machinery, such as subordinated point processes or Lévy flights, to capture the temporal correlations and memory effects often present in physical systems \cite{kobylych2017point,kozubowski2009distributional}. Furthermore, as data-driven approaches and machine learning revolutionize fields like weather forecasting \cite{lam2023}, the need for robust, physically-grounded statistical models becomes even more critical. The methods developed here can serve not as a competitor to AI, but as a vital complement: a tool for validating data-driven forecasts, providing physical priors, and rigorously characterizing the uncertainty of their predictions.

\paragraph{Data:}
Data were obtained freely from the National Institute of Meteorology \url{https://portal.inmet.gov.br} and also available available at \url{https://github.com/Raydonal/Extreme-value-BS-model}.

\paragraph{Computational Information:} Simulations were performed using the {\sf R} language and environment for statistical computing version~3.5.2.~\cite{Rmanual}, in a computer with processor Intel\textcopyright \ Core\texttrademark, i7-5500UK CPU \SI{4}{\giga\hertz}, \SI{16}{\giga\byte} RAM, System Type \SI{64}{\bit} operating system Linux. %Codes are available from the authors upon request.

\section*{Acknowledgements} 
This work was supported in part by the following Brazilian agencies: Conselho Nacional de Desenvolvimento Cient\'ifico e Tecnol\'ogico (CNPq),  No.~307626/2022-9 (AMSM), No.~303192/2022-4, 402519/2023-0  (RO)  and Fundação de Amparo a Ciência e Tecnologia do Estado da Bahia (FAPESB), No. APP0021/2023 (RO), and Coordena\c c\~ao de Aperfei\c coamento de Pessoal de N\'ivel Superior (CAPES).

% \section*{Author Contributions} JL, MB and RO: Conceptualization, Methodology, Formal Analysis, Investigation, Software, Visualization, Writing, drafting
% \& editing, Resources and Funding acquisition. All authors have contributed to interpretation of the results and manuscript revision.

\section*{Conflicts of Interest} The authors declare no conflict of interest. The funders had no role in the design of the
study; in the collection, analyzes, or interpretation of data; in the writing of the manuscript, or in the decision to
publish the results.

\appendix

\section{Appendix}\label{apendiceA}

\subsection{Proof of proposition \ref{prop:1.1.1}}
\label{pr1} The monotonicity of $f$ will be investigated by studying the sign of the first derivative of the function $g=\log(f)$. In the case where $\gamma<-1$, we have	
\begin{eqnarray*}
	g(y)
	&=&
	\log(f(y))=-\log(\alpha)+\log\left(\cosh\left(\frac{y-\eta}{2}\right)\right)-\left(\frac{1+\gamma}{\gamma}\right)\log\left(1+\frac{2\gamma}{\alpha}\sinh\left(\frac{y-\eta}{2}\right)\right)\\
	&-&
	\left(1+\frac{2\gamma}{\alpha}\sinh\left(\frac{y-\eta}{2}\right)\right)^{-1/\gamma}.
\end{eqnarray*}
Note that $g$ is derivable on the interval $(-\infty, \eta+2\mathrm{arcsinh}(-\alpha/2\gamma))$ and its first derivative can be written in the form
\begin{equation}
g'(y)=\frac{1}{2}\tanh\left(\frac{y-\eta}{2}\right)-\frac{\frac{1}{\alpha}\cosh\left(\frac{y-\eta}{2}\right)}{1+\frac{2\gamma}{\alpha}\sinh\left(\frac{y-\eta}{2}\right)}\left((1+\gamma)-\left(1+\frac{2\gamma}{\alpha}\sinh\left(\frac{y-\eta}{2}\right)\right)^{-1/\gamma}\right).
\end{equation}

Initially, let us show that $g$ is strictly increasing in $(\eta, \eta+2\mathrm{arcsinh}(-\alpha/2\gamma))$. In fact, given that $\alpha>0$, $\cosh(z)>0$ for all $z\in \R$ e $\left(1+\frac{2\gamma}{\alpha}\sinh\left(\frac{y-\eta}{2}\right)\right)>0$, we have
\[
-\frac{\frac{1}{\alpha}\cosh\left(\frac{y-\eta}{2}\right)}{1+\frac{2\gamma}{\alpha}\sinh\left(\frac{y-\eta}{2}\right)}<0.\]
On the other hand, if $\gamma<-1$, then
$$(1+\gamma)-\left(1+\frac{2\gamma}{\alpha}\sinh\left(\frac{y-\eta}{2}\right)\right)^{-1/\gamma}<0.$$
Hence, 
\[
-\frac{\frac{1}{\alpha}\cosh\left(\frac{y-\eta}{2}\right)}{1+\frac{2\gamma}{\alpha}\sinh\left(\frac{y-\eta}{2}\right)}\left((1+\gamma)-\left(1+\frac{2\gamma}{\alpha}\sinh\left(\frac{y-\eta}{2}\right)\right)^{-1/\gamma}\right)>0.\]

Hence, if $y>\eta$, we have $\tanh\left(\frac{y-\eta}{2}\right)>0$ and, consequently, $g'(y)>0$ for all $y\in(\eta, \eta+2\mathrm{arcsinh}(-\alpha/2\gamma))$. That implies that $g$ is strictly increasing over the interval $(\eta, \eta+2\mathrm{arcsinh}(-\alpha/2\gamma))$. Since the logarithmic function is strictly increasing, such property can be extended from $g$ for $f$.

\subsection{Proof of proposition \ref{prop:1.1.2}}
\label{pr2}
Under the hypothesis that  $\gamma=-1$, the first derivative $g'$ reduces to the form
\begin{equation}
g'(y)=\frac{1}{2}\frac{\sinh\left(\frac{y-\eta}{2}\right)}{\cosh\left(\frac{y-\eta}{2}\right)}+ \frac{1}{\alpha}\cosh\left(\frac{y-\eta}{2}\right).
\label{e:limgamamenos0}
\end{equation}
Therefore, when using the identity $\cosh^2(z)-\sinh^2(z)=1$, valid for all $z\in \R$, we obtain 
\begin{equation}
g'(y)=0 \Longleftrightarrow \frac{1}{\alpha}\sinh^2\left(\frac{y-\eta}{2}\right)+\frac{1}{2}\sinh\left(\frac{y-\eta}{2}\right)+\frac{1}{\alpha}=0.
\label{e:limgamamenos1}
\end{equation}
For the equation $g'(y)=0$ to have a solution, it is necessary that $\alpha\geq 4$. Under these conditions, we have the roots
$$y_1=\eta+2\mathrm{arcsinh}\left(-\frac{\alpha}{4}-\frac{1}{4}\sqrt{\alpha^2-16}\right) \   \mbox{and} \ y_2=\eta+2\mathrm{arcsinh}\left(-\frac{\alpha}{4}+\frac{1}{4}\sqrt{\alpha^2-16}\right).$$
Regarding the sign of $g'(y)$, under the assumption that $\alpha\geq 4$, we have:
\begin{enumerate}
	\item [i)]  $g'(y)<0$ when $ y_1<y<y_2$,
	\item [ii)]$g'(y)=0$ when $ y=y_1$ or $y=y_2$,
	\item [iii)]$g'(y)>0$ when $ y<y_1$ or $y>y_2$.
\end{enumerate} 
Thus, $f$ is increasing in the intervals $$(-\infty,\eta+2\mathrm{arcsinh}\left(-\frac{\alpha}{4}-\frac{1}{4}\sqrt{\alpha^2-16}\right))$$ and $$(\eta+2\mathrm{arcsinh}\left(-\frac{\alpha}{4}+\frac{1}{4}\sqrt{\alpha^2-16}\right),\eta+2{\mathrm{arcsinh}}\left(-{\alpha}/{2\gamma}\right)).$$ On the other hand, $f$ is decreasing in the interval $$(\eta+2\mathrm{arcsinh}\left(-\frac{\alpha}{4}-\frac{1}{4}\sqrt{\alpha^2-16}\right),\eta+2\mathrm{arcsinh}\left(-\frac{\alpha}{4}+\frac{1}{4}\sqrt{\alpha^2-16}\right)).$$

Now, let us approach the case $0<\alpha<2$. For all $y\in \R$, we have
\begin{equation}
-1<\frac{\sinh\left(\frac{y-\eta}{2}\right)}{\cosh\left(\frac{y-\eta}{2}\right)}<1.
\end{equation}
Consequently,
\begin{equation}
-\frac{1}{2}+\frac{1}{\alpha}\cosh\left(\frac{y-\eta}{2}\right)< \frac{1}{2}\frac{\sinh\left(\frac{y-\eta}{2}\right)}{\cosh\left(\frac{y-\eta}{2}\right)}+\frac{1}{\alpha}\cosh\left(\frac{y-\eta}{2}\right)<\frac{1}{2}+\frac{1}{\alpha}\cosh\left(\frac{y-\eta}{2}\right).
\label{ineqderivada}
\end{equation}
From (\ref{ineqderivada}) and (\ref{e:limgamamenos0}), we obtain
\begin{equation}
-\frac{1}{2}+\frac{1}{\alpha}\cosh\left(\frac{y-\eta}{2}\right)< g'(y)<\frac{1}{2}+\frac{1}{\alpha}\cosh\left(\frac{y-\eta}{2}\right).
\label{ineqsimulteglinha}
\end{equation}

On the other hand, if $0<\alpha<2$, then $0<-\frac{1}{2}+\frac{1}{\alpha}$. Also, using the fact that $\cosh\left(\frac{y-\eta}{2}\right)\geq 1$ for all $y\in \R$, we have obtained the following inequality:
\begin{equation}
0<-\frac{1}{2}+\frac{1}{\alpha}\leq -\frac{1}{2}+\frac{1}{\alpha}\cosh\left(\frac{y-\eta}{2}\right).
\label{ineqsimultegelinhacos}
\end{equation}
Thus, $g'(y)>0$ and $f$ is strictly increasing when $0<\alpha<2$ and $\gamma=-1$.

Finally, let us approach the case $2 \leq \alpha<4$. Here it is worth noting that $g'$does not change sign when  $\alpha \in (0,4)$. In fact, let us suppose by contradiction that there are $t_1$ and $t_2$, belonging to the interval $(-\infty,\eta_i +2{\mathrm{arcsinh}}\left(-{\alpha}/{2\gamma}\right))$ such that $g'(t_1)<0$ and $g'(t_2)>0$. Given that $g'$ is continuous, the intermediate value theorem assures us the existence of  $t_0 \in (-\infty,\eta_i +2{\mathrm{arcsinh}}\left(-{\alpha}/{2\gamma}\right))$ such that $g'(t_0)=0$. This is not possible, as $g'$ only admits critical points when $\alpha\geq 4$.

\subsection{Proof of proposition \ref{prop:1.1.3}} 
\label{pr3}

If $\gamma=0$, the function that assigns to each $y\in \R$ the real number $g(y)=\log(f(y))$ is derivable and its first derivative can be written in the form
\begin{equation}
g'(y)=\frac{1}{2} \frac{\sinh\left(\frac{y-\eta}{2}\right)}{\cosh\left(\frac{y-\eta}{2}\right)}-\frac{1}{\alpha}\cosh\left(\frac{y-\eta}{2}\right)\left[1-\exp\left(-\frac{2}{\alpha}\sinh\left(\frac{y-\eta}{2}\right)\right)\right].
\end{equation}
At this point, let us observe that, regardless of the value of $\alpha$, we have $g'(\eta)=0$. As far as the second derivative of $g$, we obtain
\begin{eqnarray}
g''(y)
&=&\frac{1}{4}-\frac{1}{4}\mathrm{tanh}^2\left(\frac{y-\eta}{2}\right) -\frac{1}{2\alpha}\sinh\left(\frac{y-\eta}{2}\right)+\\
&+&\frac{1}{\alpha}\exp\left(-\frac{2}{\alpha}\sinh\left(\frac{y-\eta}{2}\right)\right)\left[\frac{1}{2}\sinh\left(\frac{y-\eta}{2}\right)-\frac{1}{\alpha}\cosh^2\left(\frac{y-\eta}{2}\right)\right].
\end{eqnarray}
Hence, $g''(\eta)=\frac{1}{4}-\frac{1}{\alpha^2}$. Leading us to the following conclusions:
\begin{enumerate}
	\item [i)] If $0<\alpha<2$, then $g''(\eta)<0$.
	\item [ii)] If $\alpha=2$, then $g''(\eta)=0$.
	\item [iii)] If $\alpha>2$, then $g''(\eta)>0$.
\end{enumerate}

Therefore, $\eta$ is a local maximum of $f$, when $0<\alpha<2$.  If $\alpha>2$, then $\eta$ is a local minimum of $f$. 

Next, we will show that if $0<\alpha<2$, then $g$ is strictly increasing in the interval $(-\infty,\eta)$. In fact, for every real number $z$ we have $\cosh(z)\geq 1$. Consequently,
\begin{equation}
\frac{1}{2}\geq \frac{1}{2}\frac{1}{\cosh^2\left(\frac{y-\eta}{2}\right)}.
\label{d:meioinvcosh2}
\end{equation}

The hypothesis $0<\alpha<2$ implies
\begin{equation}
\frac{1}{2}<\frac{1}{\alpha}.
\label{d:meioalphamenorque2}
\end{equation}
From the inequalities (\ref{d:meioinvcosh2}) and (\ref{d:meioalphamenorque2}), adding the fact that the $\cosh(\cdot)$ function is strictly positive, we obtain
\begin{equation}
\frac{1}{\alpha}\cosh\left(\frac{y-\eta}{2}\right)> \frac{1}{2}\frac{1}{\cosh\left(\frac{y-\eta}{2}\right)}.
\label{d:alphameiocoshinv}
\end{equation}
Now, let us notice that when $y<\eta$, we have $-\sinh\left(\frac{y-\eta}{2}\right)>0$. Hence, from the inequality (\ref{d:alphameiocoshinv}) we have concluded that
\begin{equation}
-\frac{1}{\alpha}\sinh\left(\frac{y-\eta}{2}\right)\cosh\left(\frac{y-\eta}{2}\right)>-\frac{1}{2}\frac{\sinh\left(\frac{y-\eta}{2}\right)}{\cosh\left(\frac{y-\eta}{2}\right)}.
\label{d:invalphaprodcoshsenh}
\end{equation}

On the other hand, the assumptions $y<\eta$ and $0<\alpha<2$ entail $\sinh\left(\frac{y-\eta}{2}\right)<0$ and $-\frac{2}{\alpha}<-1$, respectively. So, 
\begin{equation}
\exp\left(-\frac{2}{\alpha}\sinh\left(\frac{y-\eta}{2}\right)\right)>\exp\left(-\sinh\left(\frac{y-\eta}{2}\right)\right).
\label{d:expxmaior1maisx}
\end{equation}
At this point, by using the fact that $e^x>1+x$ for all $x>0$ on the right-hand member of the inequality  (\ref{d:expxmaior1maisx}) we have established that
\begin{equation}
\exp\left(-\frac{2}{\alpha}\sinh\left(\frac{y-\eta}{2}\right)\right)> 1-\sinh\left(\frac{y-\eta}{2}\right)
\label{d:expsenhmaior1maissenh}
\end{equation}
By multiplying both members of the inequality (\ref{d:expsenhmaior1maissenh}) by $\frac{1}{\alpha}\cosh\left(\frac{y-\eta}{2}\right)$, we obtain
\begin{equation}
\frac{1}{\alpha}\cosh\left(\frac{y-\eta}{2}\right)\exp\left(-\frac{2}{\alpha}\sinh\left(\frac{y-\eta}{2}\right)\right)>\frac{1}{\alpha}\cosh\left(\frac{y-\eta}{2}\right)\left( 1-\sinh\left(\frac{y-\eta}{2}\right)\right).
\label{d:coshexpmaiorconhsehh}
\end{equation}
Therefore, from (\ref{d:invalphaprodcoshsenh}) and (\ref{d:coshexpmaiorconhsehh} ), we have concluded that the following inequality is valid:
\begin{equation}
\frac{1}{2}\cdot \frac{\sinh\left(\frac{y-\eta}{2}\right)}{\cosh\left(\frac{y-\eta}{2}\right)}-\frac{1}{\alpha}\cosh\left(\frac{y-\eta}{2}\right)\left[1-\exp\left(-\frac{2}{\alpha}\sinh\left(\frac{y-\eta}{2}\right)\right)\right]>0
\end{equation}

From the above, it is assured that $g$, and therefore $f$, is strictly increasing in $(-\infty,\eta)$.

\section{Inputs of the Hessian matrix of the log-likelihood function}\label{apendiceB}
The Hessian matrix $\ddL_{\hatttheta}$, for the case $\gamma=0$, has the form 
\begin{equation}
\ddL_{\hatttheta}=
{\left(
	\begin{array}{ll}
	\ddL_{11}&\ddL_{12}\\
	\ddL_{21}&\ddL_{22}\\
	\end{array}
	\right),}_{(p+1)\times(p+1)}
\end{equation}
where each component block matrix is given by
\begin{eqnarray*}
	\ddL_{11}&=&
	\left({{\frac{\partial^2 l (\ttheta )}{\partial \beta_k \partial \beta_j}}} \left|_{\bbeta={\hatbbeta},\alpha=\hat{\alpha}}\right. \right)=  -\frac{1}{4}\displaystyle \sum_{i=1}^{n} x_{ij}x_{ik}\left({\hat{\xi}_{i1}}^2\exp(-\hat{\xi}_{i2})\right)+\\
	&-&\frac{1}{4}\displaystyle \sum_{i=1}^{n}x_{ij}x_{ik}\left[\hat{\xi}_{i2}(1-\exp(-\hat{\xi}_{i2})+1 -\left(\frac{\hat{\xi}_{i1}}{\hat{\xi}_{i2}}\right)^2\right],
	\mbox{ for }  j,k=0,...,p-1.
\end{eqnarray*}

\begin{eqnarray*}
	\ddL_{12}&=&\left({{\frac{\partial^2 l (\ttheta )}{\partial \alpha \partial \beta_j}}} \left|_{\bbeta={\hatbbeta},\alpha=\hat{\alpha}}\right. \right)= -\frac{1}{2\hat{\alpha}} \displaystyle \sum_{i=1}^{n} x_{ij}\hat{\xi}_{i1}\left[1-\exp(-\hat{\xi}_{i2})\cdot (1-\hat{\xi}_{i2})\right],
\end{eqnarray*}
with $j=0,...,p-1$.
\begin{eqnarray*}
	\ddL_{21}&=&\ddL_{12}^{\top}.
\end{eqnarray*}

\begin{eqnarray*}
	\ddL_{22}&=&\left({{\frac{\partial^2 l (\ttheta )}{\partial \alpha \partial \alpha}}} \left|_{\bbeta={\hatbbeta},\alpha=\hat{\alpha}}\right. \right)= \frac{n}{\hat{\alpha}^2}-\frac{2}{\hat{\alpha}^2} \displaystyle \sum_{i=1}^{n}\hat{\xi}_{i2}\left(1-\exp(-\hat{\xi}_{i2})\right) -\frac{1}{\hat{\alpha}^2}\displaystyle \sum_{i=1}^{n}\hat{\xi}_{i2}^{2}\exp(-\hat{\xi}_{i2}).
\end{eqnarray*}

In the case under $\gamma\neq 0$, we have
\begin{equation}
\ddL_{\hatttheta}=
{\left(
	\begin{array}{lll}
	\ddL_{11}&\ddL_{12}&\ddL_{13}\\
	\ddL_{21}&\ddL_{22}&\ddL_{23}\\
	\ddL_{31}&\ddL_{32}&\ddL_{33}\\
	\end{array}
	\right)}_{(p+2)\times(p+2)}.
\end{equation}
where the entries of the matrices that make up $\ddL_{\hatttheta}$ are given by
\begin{eqnarray*}
	\ddL_{11}&=& \left({{\frac{\partial^2 l (\ttheta )}{\partial \beta_k \partial \beta_j}}} \left|_{\bbeta={\hatbbeta},\alpha=\hat{\alpha},\gamma=\hat\gamma}\right. \right)=  \frac{1}{4}\displaystyle \sum_{i=1}^{n} x_{ij}x_{ik}\left[1-\left(\frac{\hat{\xi}_{12}}{\hat{\xi}_{i1}}\right)^2\right]+\\&+& \left(-\frac{1}{4}\right)\displaystyle \sum_{i=1}^{n} x_{ij}x_{ik}\left(\frac{\hat{\xi}_{i2}}{1+\hat{\gamma}\hat{\xi}_{i2}}\right)\left[1+\hat{\gamma}-\left(1+ \hat{\gamma} \hat{\xi}_{i2}\right)^{-1/\hat{\gamma}}\right]+ \\
	&+&\frac{1}{4}\displaystyle\sum_{i=1}^{n} x_{ij}x_{ik}\left(\frac{\hat{\xi}_{i1}}{1+\hat{\gamma}\hat{\xi}_{i2}}\right)^2(1+\hat{\gamma})\left[\hat{\gamma}-\left(1+ \hat{\gamma} \hat{\xi}_{i2}\right)^{-1/\hat{\gamma}}\right],
\end{eqnarray*}
where   $j,k=0,...,p-1$.
\begin{eqnarray*}
	\ddL_{12} &=& \left({{\frac{\partial^2 l (\ttheta )}{\partial \alpha \partial \beta_j}}} \left|_{\bbeta={\hatbbeta},\alpha=\hat{\alpha},\gamma=\hat\gamma}\right. \right)= -\frac{1}{2\hat{\alpha}}\displaystyle \sum_{i=1}^{n} x_{ij}\left(\frac{\hat{\xi}_{i1}}{1+\hat{\gamma}\hat{\xi}_{i2}}\right)\left[1+\hat{\gamma}-\left(1+ \hat{\gamma} \hat{\xi}_{i2}\right)^{-1/\hat{\gamma}}\right]+ \\
	&+&\frac{1}{2\hat{\alpha}}\displaystyle\sum_{i=1}^{n} x_{ij}\frac{\hat{\xi}_{i1}\hat{\xi}_{i2}}{(1+\hat{\gamma}\hat{\xi}_{i2})^2}(1+\hat{\gamma})\left[\hat{\gamma}-\left(1+ \hat{\gamma} \hat{\xi}_{i2}\right)^{-1/\hat{\gamma}}\right], \mbox{ where }  j=0,...,p-1.
\end{eqnarray*}

\begin{eqnarray*}
	\ddL_{13}& =& \left({{\frac{\partial^2 l (\ttheta )}{\partial \gamma \partial \beta_j}}} \left|_{\bbeta={\hatbbeta},\alpha=\hat{\alpha},\gamma=\hat\gamma}\right. \right)= -\frac{1}{2}\displaystyle\sum_{i=1}^{n} x_{ij}\frac{\hat{\xi}_{i1}\hat{\xi}_{i2}}{(1+\hat{\gamma}\hat{\xi}_{i2})^2}(1+\hat{\gamma})\left[1-\frac{1}{\hat{\gamma}}\left(1+{\hat{\gamma}} \hat{\xi}_{i2}\right)^{-1/\hat{\gamma}}\right]+ \\
	&+&\frac{1}{2}\displaystyle \sum_{i=1}^{n} x_{ij}\left(\frac{\hat{\xi}_{i1}}{1+\hat{\gamma}\hat{\xi}_{i2}}\right)\left[1-\frac{1}{\hat{\gamma}^2}\log(1+\hat{\gamma}\hat{\xi}_{i2})\cdot \left(1+ \hat{\gamma} \hat{\xi}_{i2}\right)^{-1/\hat{\gamma}}\right],  j=0,...,p-1.
\end{eqnarray*}

\begin{eqnarray*}
	\ddL_{22}&=& \left({{\frac{\partial^2 l (\ttheta )}{\partial \alpha \partial \alpha}}} \left|_{\bbeta={\hatbbeta},\alpha=\hat{\alpha},\gamma=\hat\gamma}\right. \right)=  \frac{n}{\hat{\alpha}^2}-\frac{2}{\hat{\alpha}^2}\displaystyle \sum_{i=1}^{n} \left(\frac{\hat{\xi}_{i2}}{1+\hat{\gamma}\hat{\xi}_{i2}}\right)\left[1+\hat{\gamma}-\left(1+ \hat{\gamma} \hat{\xi}_{i2}\right)^{-1/\hat{\gamma}}\right]+\\
	&+&\frac{1}{\hat{\alpha}^2}\displaystyle\sum_{i=1}^{n} \left(\frac{\hat{\xi}_{i2}}{1+\hat{\gamma}\hat{\xi}_{i2}}\right)^2(1+\hat{\gamma})\left[\hat{\gamma}-\left(1+ \hat{\gamma} \hat{\xi}_{i2}\right)^{-1/\hat{\gamma}}\right].
\end{eqnarray*}

\begin{eqnarray*}
	\ddL_{23}& =& \left({{\frac{\partial^2 l (\ttheta )}{\partial \gamma \partial \alpha}}} \left|_{\bbeta={\hatbbeta},\alpha=\hat{\alpha},\gamma=\hat\gamma}\right. \right)= -\frac{1}{\hat{\alpha}}\displaystyle\sum_{i=1}^{n} \left(\frac{\hat{\xi}_{i2}}{1+\hat{\gamma}\hat{\xi}_{i2}}\right)^2(1+\hat{\gamma})\left[1-\frac{1}{\hat{\gamma}}\left(1+{\hat{\gamma}} \hat{\xi}_{i2}\right)^{-1/\hat{\gamma}}\right]+\\ 
	&+&\frac{1}{\hat{\alpha}}\displaystyle \sum_{i=1}^{n}\left(\frac{\hat{\xi}_{i2}}{1+\hat{\gamma}\hat{\xi}_{i2}}\right)\left[1-\frac{1}{\hat{\gamma}^2}\log(1+\hat{\gamma}\hat{\xi}_{i2})\cdot \left(1+ \hat{\gamma} \hat{\xi}_{i2}\right)^{-1/\hat{\gamma}}\right].
\end{eqnarray*}

\begin{eqnarray*}
	\ddL_{33} &=& \left({{\frac{\partial^2 l (\ttheta )}{\partial \gamma \partial \gamma}}} \left|_{\bbeta={\hatbbeta},\alpha=\hat{\alpha},\gamma=\hat\gamma}\right. \right)= \displaystyle \sum_{i=1}^{n}\left(\frac{\hat{\xi}_{i2}}{1+\hat{\gamma}\hat{\xi}_{i2}}\right)^2 +\\ 
	&+&\displaystyle \sum_{i=1}^{n}\left[\frac{1}{\hat{\gamma}}\left(\frac{\hat{\xi}_{i2}}{1+\hat{\gamma}\hat{\xi}_{i2}}\right)^2 + \frac{1}{\hat{\gamma}^2}\left(\frac{\hat{\xi}_{i2}}{1+\hat{\gamma}\hat{\xi}_{i2}}\right)-\frac{2}{\hat{\gamma}^3}\log(1+\hat{\gamma}\hat{\xi}_{i2})\right]g(\hat{\gamma}) \\ 
	&+&\displaystyle \sum_{i=1}^{n}\left(\frac{1}{\hat{\gamma}^2}\log(1+\hat{\gamma}\hat{\xi}_{i2})-\frac{1}{\hat{\gamma}}\frac{\hat{\xi}_{i2}}{(1+\hat{\gamma}\hat{\xi}_{i2})}\right)g'(\hat{\gamma}).
\end{eqnarray*}

\begin{eqnarray*}
	\ddL_{21}={\ddL_{12}}^{\top}, \ddL_{31}={\ddL_{13}}^{\top} \mbox{ e } \ \ddL_{32}={\ddL_{23}}^{\top}.
\end{eqnarray*}
Being that $\hat{\xi}_{i1}=\frac{2}{\hat{\alpha}}\mathrm{cosh}\left(\frac{y_i-\mathbf{x}_{i}^{\top}\hatbbeta}{2}\right)$ and $\hat{\xi}_{i2}=\frac{2}{\hat{\alpha}}\mathrm{sinh}\left(\frac{y_i-\mathbf{x}_i^{\top}\hatbbeta}{2}\right),$ for each $i=1,2,...,n,$ where $g(r)=1-(1+r\hat{\xi}_{i2})^{-1/r}$ for $r\neq 0.$


\begin{thebibliography}{10}
	
	\bibitem{solli2007}
	D.~R. Solli, C.~Ropers, P.~Koonath, and B.~Jalali.
	\newblock Optical rogue waves.
	\newblock {\em Nature}, 450(7172):1054--1057, 2007.
	
	\bibitem{dudley2008}
	John~M. Dudley, Go{\"e}ry Genty, and Benjamin~J. Eggleton.
	\newblock Harnessing and control of optical rogue waves in supercontinuum
	generation.
	\newblock {\em Optics Express}, 16(6):3644--3651, 2008.
	
	\bibitem{lima2017}
	Bismarck~C. Lima, Pablo I.~R. Pincheira, Ernesto~P. Raposo, Leonardo de~S.
	Menezes, Cid~B. de~Ara\'ujo, Anderson S.~L. Gomes, and Raman Kashyap.
	\newblock Extreme-value statistics of intensities in a cw-pumped random fiber
	laser.
	\newblock {\em Phys. Rev. A}, 96:013834, Jul 2017.
	
	\bibitem{hammani2008}
	Kamal Hammani, Bertrand Kibler, Christophe Finot, and Antonio Picozzi.
	\newblock Emergence of rogue waves from optical turbulence.
	\newblock {\em Physics Letters A}, 374(34):3585--3589, 2010.
	
	\bibitem{dudley2014}
	J.~M. Dudley, F.~Dias, M.~Erkintalo, and G.~Genty.
	\newblock Instabilities, breathers and rogue waves in optics.
	\newblock {\em Nature Photonics}, 8(10):755--764, 2014.
	
	\bibitem{onorato2013}
	M.~Onorato, S.~Residori, U.~Bortolozzo, A.~Montina, and F.~T. Arecchi.
	\newblock Rogue waves and their generating mechanisms in different physical
	contexts.
	\newblock {\em Physics Reports}, 528(2):47--89, 2013.
	
	\bibitem{Muzy-1}
	J.~F. Muzy, E.~Bacry, and A.~Kozhemyak.
	\newblock Extreme values and fat tails of multifractal fluctuations.
	\newblock {\em Phys. Rev. E}, 73:066114, Jun 2006.
	
	\bibitem{Muzy-2}
	Jean-Fran\c{c}ois Muzy and Emmanuel Bacry.
	\newblock Multifractal stationary random measures and multifractal random walks
	with log infinitely divisible scaling laws.
	\newblock {\em Phys. Rev. E}, 66:056121, Nov 2002.
	
	\bibitem{hueso2006}
	R.~Hueso and A.~S{\'a}nchez-Lavega.
	\newblock Methane storms on saturn’s moon titan.
	\newblock {\em Nature}, 442(7101):428--431, 2006.
	
	\bibitem{lobeto2021}
	Hector Lobeto, Melisa Menendez, and I{\~n}igo~J. Losada.
	\newblock Future behavior of wind wave extremes due to climate change.
	\newblock {\em Scientific Reports}, 11(1):7869, 2021.
	
	\bibitem{lam2023}
	Remi Lam, Alvaro Sanchez-Gonzalez, Matthew Willson, Peter Wirnsberger, Meire
	Fortunato, Ferran Alet, Suman Ravuri, Timo Ewalds, Zach Eaton-Rosen, Weihua
	Hu, et~al.
	\newblock Learning skillful medium-range global weather forecasting.
	\newblock {\em Science}, 382(6677):1416--1421, 2023.
	
	\bibitem{onorato2001}
	M.~Onorato, A.~R. Osborne, M.~Serio, and S.~Bertone.
	\newblock Freak waves in random oceanic sea states.
	\newblock {\em Physical Review Letters}, 86(25):5831--5834, 2001.
	
	\bibitem{janssen2003}
	Peter A. E.~M. Janssen.
	\newblock Nonlinear four-wave interactions and freak waves.
	\newblock {\em Journal of Physical Oceanography}, 33(4):863--884, 2003.
	
	\bibitem{chabchoub2011}
	A.~Chabchoub, N.~P. Hoffmann, and N.~Akhmediev.
	\newblock Rogue wave observation in a water wave tank.
	\newblock {\em Physical Review Letters}, 106(20):204502, 2011.
	
	\bibitem{bechhoefer2007}
	John Bechhoefer and Brandon Marshall.
	\newblock How \textit{Xenopus laevis} replicates dna reliably even though its
	origins of replication are located and initiated stochastically.
	\newblock {\em Physical Review Letters}, 98(9):098105, 2007.
	
	\bibitem{sornette2006critical}
	Didier Sornette.
	\newblock {\em Critical Phenomena in Natural Sciences: Chaos, Fractals,
		Selforganization and Disorder: Concepts and Tools}.
	\newblock Springer Science \& Business Media, Berlin, Heidelberg, second
	edition, 2006.
	
	\bibitem{sosa2019}
	W.~Sosa-Correa, R.~M. Pereira, A.~M.~S. Mac{\^e}do, E.~P. Raposo, D.~S.~P.
	Salazar, and G.~L. Vasconcelos.
	\newblock Emergence of skewed non-{Gaussian} distributions of velocity
	increments in isotropic turbulence.
	\newblock {\em Physical Review Fluids}, 4(6):064602, 2019.
	
	\bibitem{Maloney_2010}
	Nicholas~R. Moloney and Jörn Davidsen.
	\newblock Extreme value statistics in the solar wind: An application to
	correlated lévy processes.
	\newblock {\em Journal of Geophysical Research: Space Physics},
	115(A10):A10114, 2010.
	
	\bibitem{Schumann_2012}
	Aicko~Yves Schumann, Nicholas Moloney, and Joern Davidsen.
	\newblock Extreme value and record statistics in heavy-tailed processes with
	long-range memory.
	\newblock {\em AGU Geophysical Monograph 192: Extreme Events and Natural
		Hazards - The Complexity Perspective}, pages 315--, 01 2012.
	
	\bibitem{beirlant2006}
	Jan Beirlant, Yuri Goegebeur, Jozef Teugels, and Johan Segers.
	\newblock {\em Statistics of Extremes: Theory and Applications}.
	\newblock John Wiley \& Sons, 2004.
	
	\bibitem{coles2001introduction}
	Stuart Coles.
	\newblock {\em An Introduction to Statistical Modeling of Extreme Values}.
	\newblock Springer, London, 2001.
	
	\bibitem{gomes2015}
	M.~Ivette Gomes and Armelle Guillou.
	\newblock Extreme value theory and statistics of univariate extremes: a review.
	\newblock {\em International Statistical Review}, 83(2):263--292, 2015.
	
	\bibitem{ferreira2012}
	M.~Ferreira, M.~I. Gomes, and V{\'\i}ctor Leiva.
	\newblock On an extreme value version of the {Birnbaum-Saunders} distribution.
	\newblock {\em REVSTAT - Statistical Journal}, 10(2):181--210, 2012.
	
	\bibitem{leiva2016extreme}
	V{\'\i}ctor Leiva, Marta Ferreira, M.~Ivette Gomes, and Camilo Lillo.
	\newblock Extreme value {Birnbaum-Saunders} regression models applied to
	environmental data.
	\newblock {\em Stochastic Environmental Research and Risk Assessment},
	30(3):1045--1058, 2016.
	
	\bibitem{birnbaum1969}
	Z.~W. Birnbaum and S.~C. Saunders.
	\newblock A new family of life distributions.
	\newblock {\em Journal of Applied Probability}, 6(2):319--327, 1969.
	
	\bibitem{cook1986}
	R.~Dennis Cook.
	\newblock Assessment of local influence.
	\newblock {\em Journal of the Royal Statistical Society: Series B
		(Methodological)}, 48(2):133--155, 1986.
	
	\bibitem{nguyen2017}
	H.~Chau Nguyen, Riccardo Zecchina, and Johannes Berg.
	\newblock Inverse statistical problems: from the inverse {Ising} problem to
	data science.
	\newblock {\em Advances in Physics}, 66(3):197--261, 2017.
	
	\bibitem{poon1999conformal}
	Yat-Sum Poon and Wai-Yin Poon.
	\newblock Conformal normal curvature and assessment of local influence.
	\newblock {\em Journal of the Royal Statistical Society: Series B (Statistical
		Methodology)}, 61(1):51--61, 1999.
	
	\bibitem{embrechts1999extreme}
	Paul Embrechts, Sidney~I. Resnick, and Gennady Samorodnitsky.
	\newblock Extreme value theory as a risk management tool.
	\newblock {\em North American Actuarial Journal}, 3(2):30--41, 1999.
	
	\bibitem{dey2016extreme}
	Dipak~K. Dey and Jun Yan, editors.
	\newblock {\em Extreme Value Modeling and Risk Analysis: Methods and
		Applications}.
	\newblock CRC Press, 2016.
	
	\bibitem{gomes2015extreme}
	M.~Ivette Gomes and Armelle Guillou.
	\newblock Extreme value theory and statistics of univariate extremes: a review.
	\newblock {\em International Statistical Review}, 83(2):263--292, 2015.
	
	\bibitem{katz1992extreme}
	Richard~W. Katz and Barbara~G. Brown.
	\newblock Extreme events in a changing climate: variability is more important
	than averages.
	\newblock {\em Climatic Change}, 21(3):289--302, 1992.
	
	\bibitem{leclerc2007non}
	Martin Leclerc and Taha B. M.~J. Ouarda.
	\newblock Non-stationary regional flood frequency analysis at ungauged sites.
	\newblock {\em Journal of Hydrology}, 343(3-4):254--265, 2007.
	
	\bibitem{cooley2009extreme}
	Daniel Cooley.
	\newblock Extreme value analysis and the study of climate change.
	\newblock {\em Climatic Change}, 97(1-2):77--83, 2009.
	
	\bibitem{mises1936distribution}
	R.~von Mises.
	\newblock La distribution de la plus grande de $n$ valeurs.
	\newblock {\em Revue Math{\'e}matique de l'Union Interbalkanique}, 1:141--160,
	1936.
	
	\bibitem{jenkinson1955frequency}
	Arthur~F. Jenkinson.
	\newblock The frequency distribution of the annual maximum (or minimum) values
	of meteorological elements.
	\newblock {\em Quarterly Journal of the Royal Meteorological Society},
	81(348):158--171, 1955.
	
	\bibitem{fisher1928limiting}
	R.~A. Fisher and L.~H.~C. Tippett.
	\newblock Limiting forms of the frequency distribution of the largest or
	smallest member of a sample.
	\newblock In {\em Mathematical Proceedings of the Cambridge Philosophical
		Society}, volume~24, pages 180--190. Cambridge University Press, 1928.
	
	\bibitem{Dombry2015}
	Cl{\'e}ment Dombry.
	\newblock Existence and consistency of the maximum likelihood estimators for
	the extreme value index within the block maxima framework.
	\newblock {\em Bernoulli}, 21(1):420--436, 2015.
	
	\bibitem{Bucher2017}
	A.~Bucher and J.~Segers.
	\newblock On the maximum likelihood estimator for the generalized extreme-value
	distribution.
	\newblock {\em Extremes}, 20(4):839--872, 2017.
	
	\bibitem{DombryFerreira2019}
	Cl{\'e}ment Dombry and Ana Ferreira.
	\newblock Maximum likelihood estimators based on the block maxima method.
	\newblock {\em Bernoulli}, 25(3):1690--1723, 2019.
	
	\bibitem{Rieck1991}
	James~R. Rieck and Jerry~R. Nedelman.
	\newblock A log-linear model for the {Birnbaum-Saunders} distribution.
	\newblock {\em Technometrics}, 33(1):51--60, 1991.
	
	\bibitem{poonpoon2010}
	Yat-Sun Poon and Wai-Yin Poon.
	\newblock {\em Application of Elementary Differential Geometry to Influence
		Analysis}.
	\newblock Higher Education Press, Beijing, 2010.
	
	\bibitem{pachter2024entropy}
	Jonathan~Asher Pachter, Ying-Jen Yang, and Ken~A. Dill.
	\newblock Entropy, irreversibility and inference at the foundations of
	statistical physics.
	\newblock {\em Nature Reviews Physics}, 6(6):382--393, 2024.
	
	\bibitem{nguyen2017inverse}
	H.~Chau Nguyen, Riccardo Zecchina, and Johannes Berg.
	\newblock Inverse statistical problems: from the inverse {Ising} problem to
	data science.
	\newblock {\em Advances in Physics}, 66(3):197--261, 2017.
	
	\bibitem{jacobs2010stochastic}
	Kurt Jacobs.
	\newblock {\em Stochastic Processes for Physicists: Understanding Noisy
		Systems}.
	\newblock Cambridge University Press, 2010.
	
	\bibitem{kobylych2017point}
	K.~Kobylych and L.~Sakhno.
	\newblock Point processes subordinated to compound {Poisson} processes.
	\newblock {\em Theory of Probability and Mathematical Statistics}, 94:89--96,
	2017.
	
	\bibitem{kozubowski2009distributional}
	Tomasz~J. Kozubowski and Krzysztof Podg{\'o}rski.
	\newblock Distributional properties of the negative binomial {L{\'e}vy}
	process.
	\newblock {\em Probability and Mathematical Statistics}, 29(1):43--71, 2009.
	
	\bibitem{Rmanual}
	{R Core Team}.
	\newblock {\em {\sf R}: A Language and Environment for Statistical Computing}.
	\newblock R Foundation for Statistical Computing, Vienna, Austria, 2025.
	
\end{thebibliography}
\end{document}